\documentclass[acmsmall,nonacm]{acmart}
\usepackage[utf8]{inputenc}

\usepackage{ulem}

\newcommand{\FF}{\mathbb{F}}

\usepackage{subcaption}

\usepackage[noend]{algpseudocode}
\usepackage{amsmath,amsfonts,amsthm,algorithm}
\usepackage{wrapfig}
\usepackage{graphicx}
\usepackage{textcomp}
\usepackage{xcolor}
\usepackage{tikz}
\usepackage{multicol}
\usepackage{url,subcaption,svg,tcolorbox,arydshln}
\usetikzlibrary{decorations.pathreplacing}

\usepackage[table]{xcolor}
\definecolor{lightgray}{gray}{0.96}
\usepackage{framed}

\begin{document}

\title{Byzantine Consensus in Directed Graphs\\
with Message Authentication}

\author{Nitin H. Vaidya}
\email{nitin.vaidya@georgetown.edu}
\affiliation{%
  \institution{Georgetown University}
  \country{Washington DC, USA}
}

\author{Lewis Tseng}
\affiliation{%
  \institution{UMass Lowell}
  \country{USA}}
\email{lewistseng@acm.org}

% \renewcommand{\shortauthors}{Vaidya and Tseng}

\begin{abstract}
~\\~\\
\noindent
  ABSTRACT:\footnote{Revised May 2026. This is a  revised version of a manuscript prepared in February 2026. In particular, Lemma \ref{l_component_size} and Table \ref{t:undirected} are added, the Related Work section is revised, and some other minor edits have been made.}
  
  We consider the problem of reaching consensus in communication networks that are modeled by
directed graphs. We assume the existence of a message authentication mechanism (such as digital signatures)
to verify the integrity of messages. We identify the necessary and sufficient conditions on the {\it directed} communication graph for the following problems to be solvable: (i) exact consensus in synchronous systems; and (ii) approximate consensus in
asynchronous systems.
\end{abstract}

\maketitle

\thispagestyle{empty} 

\newpage 
\setcounter{page}{1}

\setlength{\parsep}{2pt}
\section{Introduction}
\label{s:intro}

Byzantine consensus is an important distributed primitive, and over forty years of research has been conducted on this problem \cite{lamport_agreement,lamport_agreement2,dolev_82_BG}. In recent years, there has been work on Byzantine consensus in systems where the underlying network is modeled by a directed graph (e.g., \cite{Tseng_podc2015,Tseng_PODC20, Sundaram_journal,Ding2021RelayIABC}). Although the necessary and sufficient conditions for Byzantine consensus {\it in the absence} of message authentication in directed graphs are known \cite{Tseng_podc2015,Tseng_PODC20}, the necessary and sufficient conditions under message authentication were not previously established for directed graphs. This paper addresses this gap. We consider a distributed system under the following assumptions:
\begin{itemize}

    \item The communication network is modeled by a directed graph $G(V,E)$ where each vertex in $V$ represents a node in the distributed  system, and each directed edge in $E$ represents a directed communication link. In particular, a node $u\in V$ can send messages directly to a node $v\in V$ if and only if $(u,v)\in E$. We will use the terms {\it node} and {\it vertex} interchangeably.
    \item Let us define $n=|V|$, the number of nodes in the system.
    
    \item Up to $f$ of the $n$ nodes may be Byzantine faulty. We assume that $n>f>0$.
    
    \item Each node $u \in V$ has a binary input (i.e., 0 or 1) denoted as $x_u$.

    \item Each node $u\in V$ produces an output denoted as $y_u$.

     \item A message authentication mechanism is available -- in particular, we assume the availability of digital signatures using public-key cryptography \cite{lamport_agreement2}. With message authentication, a faulty node cannot tamper a message from a non-faulty node without the tampering  being detectable.
\end{itemize}
In a synchronous system, a Byzantine consensus algorithm \cite{lamport_agreement,lamport_agreement2} should satisfy the requirements below.
\begin{itemize}
    \item \textit{Termination}: each non-faulty node eventually outputs a value.

    \item \textit{Validity}: the output at each non-faulty node must equal the input of some non-faulty node.

    \item \textit{Agreement}: the outputs of the non-faulty nodes must be identical.
\end{itemize}

In an asynchronous system,
due to the FLP impossibility \cite{FLP_one_crash}, we consider $\epsilon$-approximate Byzantine consensus \cite{AA_Dolev_1986}, with $\epsilon>0$, which requires the algorithm to satisfy the  requirements below.
\begin{itemize}
    \item \textit{Termination}: each non-faulty node eventually outputs a value.

    \item \textit{Validity}: the output at each non-faulty node must be in the convex hull (or the range) of the inputs of the non-faulty nodes.

    \item \textit{$\epsilon$-agreement}: the outputs of the non-faulty nodes must be within $\epsilon$ of each other.
\end{itemize}

\section{Our Main Results}
\label{s:main}

We assume that all the nodes know the graph $G(V,E)$.
We obtain the necessary and sufficient conditions on the communication graph $G(V,E)$ for Byzantine consensus to be achievable in synchronous and asynchronous systems {\it with message authentication}, assuming $n>f>0$. We first introduce some terminology  used to present the results in this section. Additional terminology and notation will be introduced in Section \ref{s_terminology}.

Given a set $F\subset V$ with $|F|\leq f$, let $G_F$ denote the subgraph of graph $G$ induced by the nodes in set\footnote{Set subtraction: Given sets $P$ and $Q$, $P-Q=\{ x~|~x\in P, ~~x\not\in Q\}$. For example, if $P=\{a,b,c\}$ and $Q=\{b,c,d\}$, then $P-Q=\{a\}$.} $V-F$. Thus, $V-F$ is the set of vertices in $G_F$, and the set of edges in $G_F$ is the subset of those edges in $E$ that have  both the endpoints in $V-F$. Thus, the set of edges in $G_F$ is $\{(u,v)~|~u,v\in V-F~\mbox{~and~}~ (u,v)\in E\}$.

Now we define ``reach sets'', and three  ``reach conditions''. The definition of {\it reach sets} uses paths in graph $G_F$ defined above. Reach sets were introduced in prior work \cite{Tseng_PODC20} on related consensus problems. The reach conditions below use a parameter $\rho$. Reach conditions equivalent to the case of $\rho=1$ were also introduced in the prior work \cite{Tseng_PODC20}. We will summarize the results in \cite{Tseng_PODC20} soon.

\begin{definition}[Reach set]\label{def:reach}
	Given a set $F\subset V$, $|F|\leq f$, and $u\in V-F$, define $reach_u(F)$ as follows.
	$$reach_u(F)=\{w\in V-F ~|~ \text{ $w$ has a directed path to $u$ in graph } G_{F}\}$$
\end{definition}
Thus, $reach_u(F)$ defined above is the set of nodes that have paths to node $u$ in the subgraph $G_F$ of $G(V,E)$ induced by the nodes in $V-F$. Thus, the set $reach_u(F)$ includes node $u$, but it does not include any nodes from $F$. For example, in Figure \ref{fig:1-reach}, assuming $F=\{4\}$, $reach_{3}(F)=\{1,2,3\}$ and $reach_{5}(F)=\{1,2,3,5,6\}$.

\begin{definition}[Reach Conditions]
\label{def:reach1} We define three {\it reach conditions}. Note that the $k$-reach condition below uses $k$ subsets of $V$ of size at most $f$.\footnote{In the reach conditions, $u$ and $v$ may be distinct or identical nodes. Similarly, sets $F,F_1,F_2$ may possibly have empty or non-empty pairwise intersections.}
	 \begin{itemize}
		\item \textbf{1-reach with parameter $\rho$:} For any $F\subset V$ such that $|F|\leq f$ and any $u,v\in V-F$, \vspace*{5pt} \\ \hspace*{1.7in}  $|reach_u(F)\cap reach_v(F)|\geq \rho $
		\item \textbf{2-reach with parameter $\rho$:} For any 
        subsets $F_1$ and $F_2$ of $V$ such that $|F_1|,|F_2|\leq f$, and any $u\in V-F_1$ and any $v\in V-F_2$, \vspace*{5pt}  \\ \hspace*{1.7in} $|reach_u(F_1)\cap reach_v(F_2)|\geq \rho$
		
		\item \textbf{3-reach with parameter $\rho$:}
        For any 
        subsets $F$, $F_1$ and $F_2$ of $V$ such that $|F|,|F_1|,|F_2|\leq f$, and any $u\in V-F-F_1$ and any $v\in V-F-F_2$, \vspace*{5pt} \\ \hspace*{1.7in}
		$|reach_u(F\cup F_1)\cap reach_v(F\cup F_2)|\geq \rho$
	\end{itemize}
	
\end{definition}

\subsection{Tight Conditions for Byzantine Consensus with Message Authentication}
\label{s:results}

For synchronous and asynchronous systems both, we have obtained two equivalent necessary and sufficient conditions for Byzantine consensus in directed graphs with message authentication. We present one of these conditions for synchronous and asynchronous consensus here. The alternate conditions are presented in Sections \ref{s:sync} and \ref{s:async}. All the requisite proofs are included later in the paper.

\begin{framed}
\noindent{\tt Exact Consensus in Synchronous Systems:}

\begin{theorem}
\label{t_reach_sync}
In a synchronous system with message authentication, exact Byzantine consensus in the presence of up to $f$ Byzantine faulty nodes is possible if and only if graph $G(V,E)$ satisfies the 1-reach condition with parameter $\rho=f+1$.
\end{theorem}

\noindent{\tt $\epsilon$-Approximate Consensus in Asynchronous Systems:}

\begin{theorem}
\label{t_reach_async}
In an asynchronous system with message authentication, $\epsilon$-approximate Byzantine consensus in the presence of up to $f$ Byzantine faulty nodes, with any $\epsilon>0$, is possible if and only if graph $G(V,E)$ satisfies the 2-reach condition with parameter $\rho=f+1$.
\end{theorem}
% }
\end{framed}

For comparison with prior work on consensus in directed graphs, Table \ref{t:directed} summarizes the previous results as well as our results (in the rightmost column). The previous work has obtained tight conditions for tolerating crash\footnote{For crash failures, the validity condition for exact consensus requires that the output of the non-faulty nodes should equal the input of some node, and the validity condition for $\epsilon$-approximate consensus requires that the output be in the range of the inputs of all the nodes.} failures \cite{Tseng_podc2015}, and for Byzantine failures {\it without} message authentication \cite{Tseng_podc2015,Tseng_PODC20}. 
\cite{Tseng_PODC20} also obtained reach
conditions equivalent to the graph conditions obtained in \cite{Tseng_podc2015}.
This paper obtains results for Byzantine failures {\it with} message authentication. 

\begin{table}[h]
\begin{center}
\begin{tabular}{|p{55pt}|p{75pt}|p{110pt}|p{105pt}|} \hline
& Up to $f$ \newline Crash faults
& Up to $f$ \newline Byzantine faults {\bf without} \newline   message authentication & \cellcolor{lightgray} Up to $f$\newline Byzantine faults {\bf with}\newline  message authentication\\ \hline
Synchronous \newline
Systems & {1-reach} with $\rho = 1$ \newline \hspace*{20pt}\cite{Tseng_podc2015} & { 3-reach} with $\rho = 1$  \newline \hspace*{20pt}\cite{Tseng_podc2015}
& \cellcolor{lightgray} 1-reach with $\rho=f+1$ {\bf (This paper)}\\ \hline
Asynchronous \newline
Systems & { 2-reach} with $\rho=1$ \newline\hspace*{20pt}\cite{Tseng_podc2015} & { 3-reach} with $\rho=1$ \newline \hspace*{20pt}\cite{Tseng_PODC20}
& \cellcolor{lightgray} 2-reach with $\rho=f+1$ {\bf (This paper)}
\\ \hline
\end{tabular}
\end{center}
\caption{Necessary and sufficient conditions for consensus in directed graphs}
\label{t:directed}
\end{table}

It is easy to show the following total order on the above graph conditions, where $P \Leftarrow Q$ means that condition $Q$ implies condition $P$ (or, equivalently, condition $Q$ is stronger than $P$). Appendix \ref{a:reach} proves this order. Shown in bold below are the conditions in the rightmost column in Table \ref{t:directed}. The order below is strict: in particular, if $P\Leftarrow Q$ below then there exists a graph that satisfies condition $P$ but not condition $Q$.

\begin{tabular}{p{50pt}p{10pt}p{60pt}p{10pt}p{50pt}p{10pt}p{60pt}p{10pt}p{50pt}}
1-reach \newline with $\rho=1$ &
$\Leftarrow$ & {\bf 1-reach\newline with $\rho=f+1$} &
$\Leftarrow$ & 2-reach \newline with $\rho=1$ & 
$\Leftarrow$ & {\bf 2-reach \newline with $\rho=f+1$} &
$\Leftarrow$ & 3-reach \newline with $\rho=1$ 
\end{tabular}

For $f=1$, Figure \ref{fig:1-reach} provides an example of a graph that satisfies the 1-reach condition with $\rho=f+1=2$ (but does not satisfy 2-reach with $\rho=1$), and Figure \ref{fig:2-reach} provides an example of a graph that satisfies 2-reach with $\rho=f+1=2$ (but does not satisfy 3-reach with $\rho=1$).
\begin{figure}[!htb]
    \scriptsize 
    \centering
    \begin{subfigure}[b]{0.35\textwidth}
        \includegraphics[height=3.75cm,width=.8\linewidth]{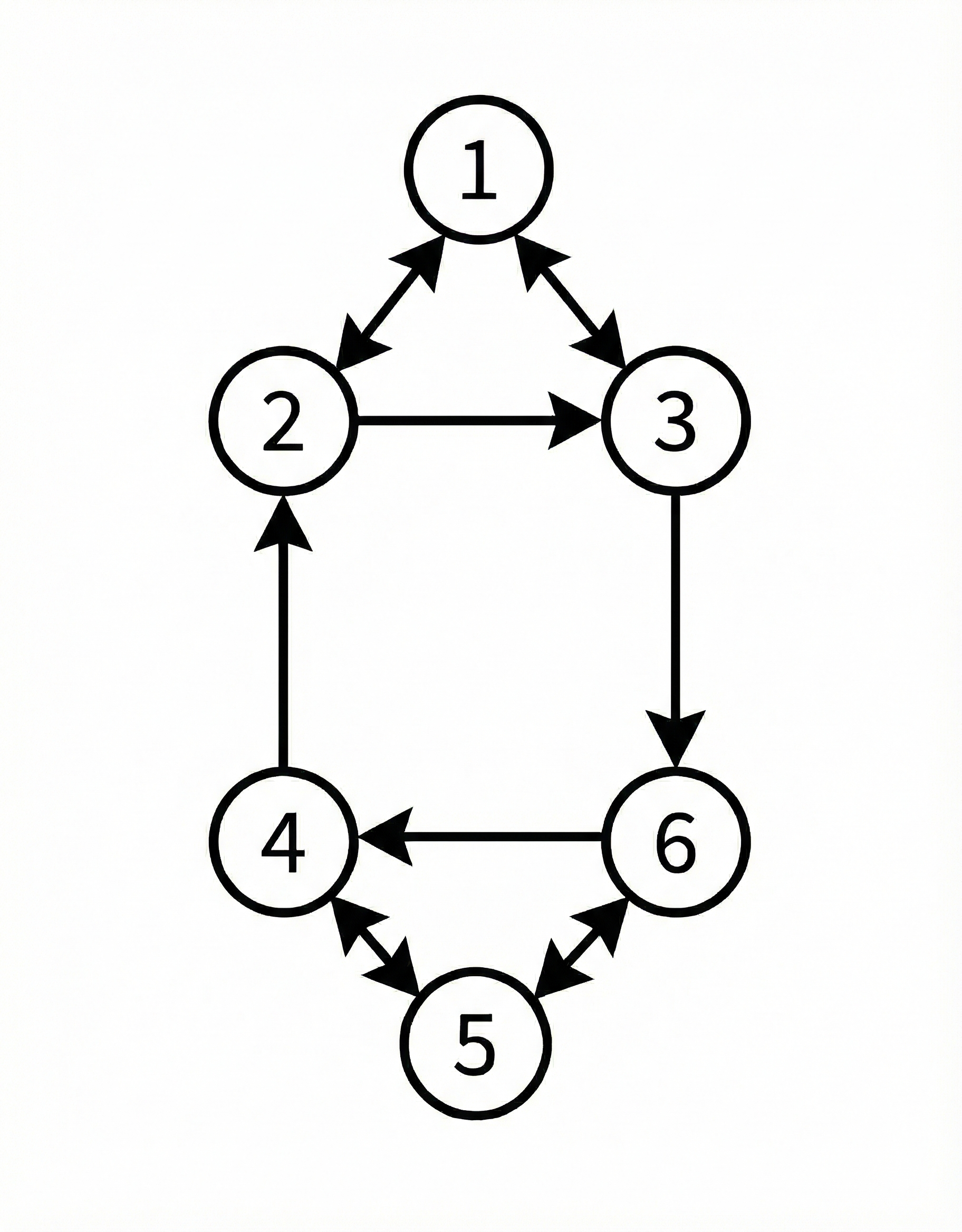}
        \caption{1-reach with $\rho=f+1=2$}
        \label{fig:1-reach}
    \end{subfigure} \hfill
    \begin{subfigure}[b]{0.35\textwidth}
        \includegraphics[height=3.75cm,width=.8\linewidth]{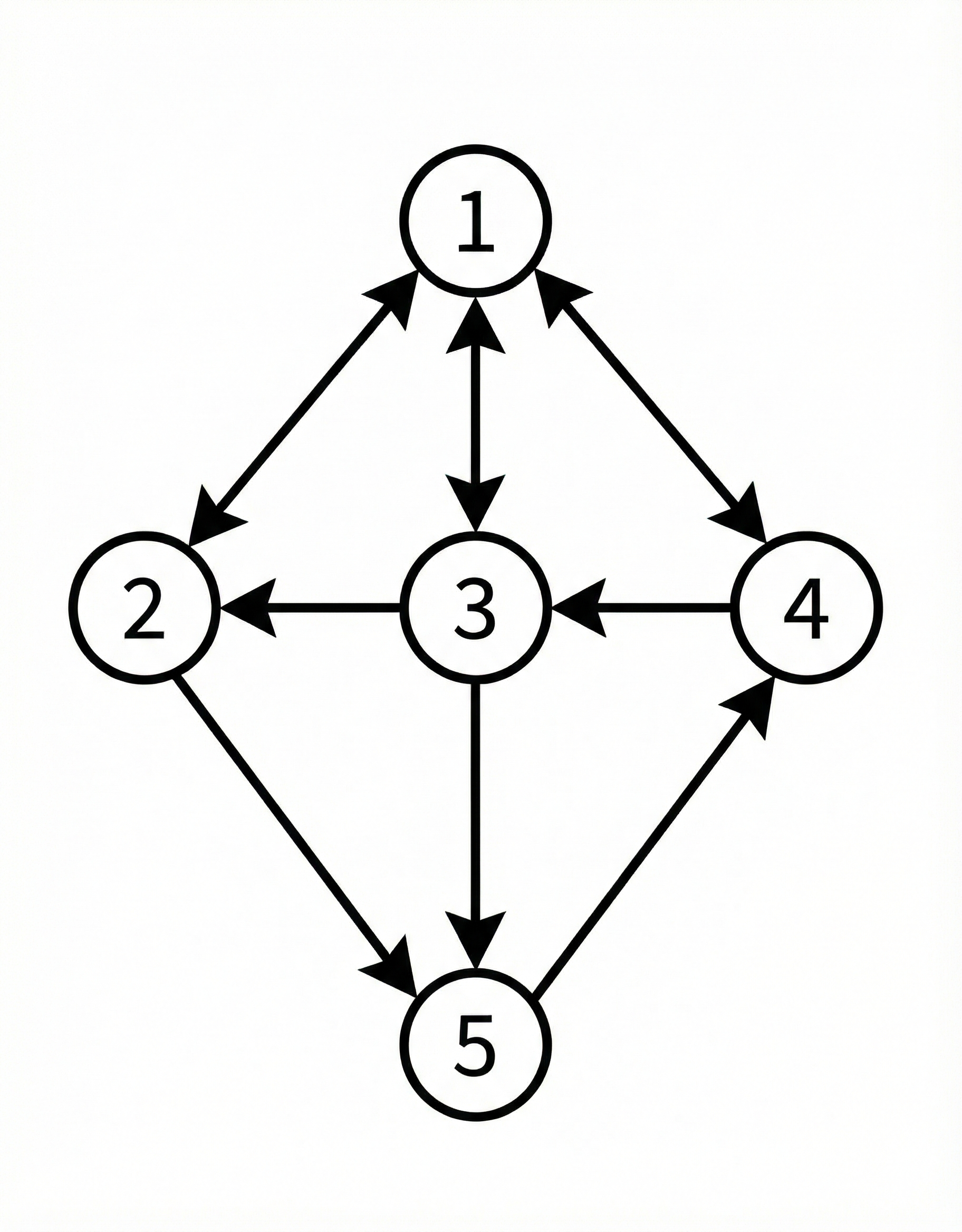}
        \caption{2-reach with $\rho=f+1=2$}
        \label{fig:2-reach}
    \end{subfigure}\hfill
    \caption{Example graphs for $f=1$ }
    \label{fig:example}
\end{figure}

For comparison, we summarize in Table \ref{t:undirected} the necessary and sufficient conditions for undirected graphs, where $\kappa$ denotes the vertex connectivity of graph $G(V,E)$ (assuming that the graph is undirected). Not surprisingly, for undirected graphs, the conditions listed in Table \ref{t:directed} are equivalent to the corresponding conditions in Table \ref{t:undirected}.

\begin{table}[h]
\begin{center}
\begin{tabular}{|p{55pt}|p{75pt}|p{110pt}|p{105pt}|} \hline
& Up to $f$ \newline Crash faults & Up to $f$ \newline Byzantine faults {\bf without} \newline   message authentication & Up to $f$\newline Byzantine faults {\bf with}\newline  message authentication\\ \hline
Synchronous \newline
Systems & $n>f$ and $\kappa>f$ \newline \hspace*{20pt}\cite{welch_book} & $n>3f$ and $\kappa>2f$  \newline \hspace*{20pt}\cite{dolev_82_BG}
&  $n>2f$ and $\kappa>f$  \newline \hspace*{20pt}\cite{Byz_Lamport_1982,DolevS83_AuthByz,welch_book} \\ \hline
Asynchronous \newline
Systems & $n>2f$ and $\kappa>f$ \newline\hspace*{20pt}\cite{welch_book,AA_Fekete_aoptimal} & $n>3f$ and $\kappa>2f$ \newline \hspace*{20pt}\cite{dolev_82_BG,AA_Dolev_1986,impossible_proof_lynch}
&  $n>3f$ and $\kappa>f$  \newline \hspace*{20pt}\cite{AA_Dolev_1986,DolevS83_AuthByz,welch_book,GhineaLW23}
\\ \hline
\end{tabular}
\end{center}
\caption{Necessary and sufficient conditions for consensus in undirected graphs. $\kappa$ denotes vertex connectivity.}
\label{t:undirected}
\end{table}

\section{Related Work}
\label{sec:related}

We summarized some of the most closely related work above. We now briefly discuss some other work on Byzantine consensus, and reliable communication in the presence of Byzantine faults.
There is work on using message authentication for exact consensus in {\it undirected} graphs \cite{lamport_agreement2,Bansal_disc11} as well as {\it approximate} Byzantine consensus in complete graphs \cite{DBLP:conf/podc/GhineaLW22}.

Lamport et al. \cite{lamport_agreement2} considered undirected graphs where all the nodes can add a cryptographic signature to their messages. Several past works have considered {\it partial authentication}, where only a subset of the nodes are able to add tamper-proof cryptographic signatures to the messages. Bansal et al. \cite{Bansal_disc11} considered undirected graphs under {\it partial} authentication (\cite{cryptoeprint:2008/287} considered a related problem in complete graphs). More recently, for {\it complete graphs}, Lenzen, Loss, Shi, and Wagner~\cite{LenzenLSW26PartialAuth} have also studied Byzantine consensus with {\it partial} authentication. These past works on partial authentication differ somewhat in the assumptions they make.
In other related research, there is significant work on reliable communication with message authentication in directed graphs (e.g., \cite{DBLP:journals/ipl/KishoreVS18,Srinathan06}). Other work 
(e.g., \cite{DBLP:conf/icdcn/PatraCR10,DBLP:journals/jacm/BadanidiyuruPCSR12,PPS17}) has explored reliable message transmission without message authentication.

For complete graphs, {\it approximate} consensus with signatures has been studied in synchronous systems in \cite{Round-optimal_BA_Eurocrypt22,OptimalClockSyncSign_PODC22}. 
Ghinea, Liu-Zhang, and Wattenhofer \cite{DBLP:conf/podc/GhineaLW22} recently introduced a novel class of approximate consensus algorithms  with message authentication. Their algorithms tolerate up to $t_s<n/2$ corruptions in synchronous runs while also tolerating $t_a<n/3$ in asynchronous runs. 
Their subsequent work \cite{GhineaLW23} considers $D$-dimensional approximate consensus in {\it complete graphs}. \cite{Vaidya_BVC,Vaidya_incomplete,Mendes_STOC13} previously introduced multi-dimensional vector consensus (including exact and $\epsilon$-approximate consensus). \cite{SamirAsyncConsensusArxiv19} studies approximate consensus in incomplete undirected graphs under the {\it local broadcast} communication model. The {\it local broadcast} model used in \cite{SamirAsyncConsensusArxiv19} assumes that all outgoing neighbors of a node receive all its transmissions identically (which constrains a faulty node's ability to equivocate). The algorithms in \cite{Vaidya_BVC,Vaidya_incomplete,Mendes_STOC13,SamirAsyncConsensusArxiv19} do not use message authentication.

For directed graphs, there is work on iterative approximate consensus algorithms that use a {\it restricted algorithm structure} to achieve approximate  consensus. These algorithms, in each iteration, allow only one-hop communication using local topology information, and carry a small amount of state across iterations. While these algorithms are simple and have a low memory requirement, as a trade-off, the communication graphs need to satisfy more restrictive (i.e., stronger) properties. Tight graph conditions for iterative approximate Byzantine consensus on directed graphs (in the absence of message authentication) have been obtained in 
\cite{vaidya_PODC12,Sundaram_journal,sv17}.
Ding \cite{Ding2020RelayABC,Ding2021RelayIABC} studied iterative algorithms for approximate Byzantine consensus with message authentication, but under graph conditions that are {\it not tight}. Tight conditions for exact Byzantine consensus in directed graphs under the {\it local broadcast} communication model, and without message authentication, are identified in \cite{khan_et_al:LIPIcs.OPODIS.2019.30}.

\section{Graph Terminology and Notations}
\label{s_terminology}

Recall that, for $F\subset V$, we defined $G_F$ as the subgraph of graph $G$ induced by the vertices in $V-F$. We also defined {\it reach sets} and {\it reach conditions} in Section \ref{s:main}.
For the ease of reference, this section summarizes additional graph-related terminology and notations, which will be used later in the paper.

\begin{definition}[Edge set $E_{P,Q}$]
Given $P,Q\subseteq V$, let $E_{P,Q}$ denote the subset of $E$ consisting of edges from the nodes in $P$ to the nodes in $Q$. That is, $E_{P,Q}=\{(p,q)~|~(p,q)\in E, ~p\in P,~q\in Q\}$.
\end{definition}

Recall that a {\it strongly connected component} in a given graph is a {\it maximal} set of vertices such that each pair of vertices in that set has paths to each other in the given graph. Thus, if a graph contains more than one strongly connected component, then any two distinct strongly connected components are disjoint. A non-empty graph will include at least one strongly connected component. Also, for some graph $H$, if a subset of vertices $X$ in $H$ does not have any incoming edges from vertices outside $X$, then the subgraph of $H$ induced by $X$ will contain at least one strongly connected component of graph $H$.

\begin{definition}[Source Component]
\label{def:sourceComponent}    
A {\it source component} in a given graph is a strongly connected component that does
not have incoming links from any nodes outside that component.
\end{definition}

\begin{itemize}

\item If $(u,v)\in E$, then $u$ is an {\it incoming} neighbor of $v$, and $v$ is an {\it outgoing} neighbor of $u$.
\item For $P\subseteq V$, if there exists $w\in V-P$ and $u\in P$ such that $(w,u)\in E$, then $w$ is an {\it incoming} neighbor of set $P$.

\item For $P\subseteq V$, if there exists $u\in P$ and $v\in V-P$ such that $(u,v)\in E$, then $v$ is an {\it outgoing} neighbor of set $P$.

\item A node in a directed path is said to be an {\it internal node} if it is neither the source (i.e., the first node) nor the destination (i.e., the last node) in the path. For example, in Figure \ref{fig:1-reach}, in the directed path $1\rightarrow 3 \rightarrow 6 \rightarrow 5$, nodes $3$ and $6$ are internal nodes; node $1$ is the source node, and node 5 is the destination node.

\item We will consider a bipartite graph consisting of two disjoint sets of nodes, say $P$ and $Q$, and a set of directed links from nodes in $P$ to nodes in $Q$. Suppose that the set of links is called $L$, and the bipartite graph is called $B$. Then we will denote this bipartite graph as $B(P,Q,L)$.

\end{itemize}

\section{Synchronous Byzantine Consensus with Message Authentication}
\label{s:sync}

We now present an alternate version of the necessary and sufficient conditions for exact consensus in synchronous systems. The equivalence of the alternate condition in Theorem \ref{t_authByzSync} below and the 1-reach condition with $\rho=f+1$ (used in Theorem \ref{t_reach_sync}) is proved in Appendix \ref{a:sync_equiv}.

A {\it source component} is defined in Section \ref{s_terminology}.
As defined in Section \ref{s:main}, $G_F$ denotes the subgraph of $G(V,E)$ induced by the nodes in $V-F$. Condition S below considers source components in $G_F$. This condition is named "S" because it is used in the context of {\bf S}ynchronous systems.

\begin{framed}
\begin{definition}[Condition S]
\label{def:cond-authByzSync}
$G(V,E)$ is said to satisfy Condition S if \\for any  $F \subset V$ with $|F| \leq f$, there is a unique source component in graph $G_F$, and the source component contains at least $f+1$ nodes.
\end{definition}

\begin{theorem}
\label{t_authByzSync}
In a synchronous system with message authentication, exact Byzantine consensus is possible in graph $G(V,E)$ if and only if graph $G(V,E)$ satisfies Condition S. 
\end{theorem}
\end{framed} 
\begin{proof}
Theorem \ref{t_authByzSync} is proved in two parts.
The proof of the necessity of Condition S for Byzantine consensus is presented in Appendix \ref{a:sync:necessity}. To prove the sufficiency of Condition S, we present a synchronous algorithm in Section \ref{s:sync:algo} and show its correctness. Some key results required to show the correctness of this algorithm are proved in Sections \ref{s:sync:impli} and \ref{s:sync:algo}. The rest of the proof of correctness of the synchronous algorithm is in Appendix \ref{a:sync:algo:proof}.
\end{proof}
The correctness of Theorem \ref{t_reach_sync} follows from the correctness of Theorem \ref{t_authByzSync}, and the equivalence proved between the graph conditions in these theorems (proved in Appendix \ref{a:sync_equiv}).\\
\indent Now we introduce the notation $S_F$, which is often used in the discussion that follows. 
\begin{framed}
\begin{definition}[Unique source component $S_F$]
\label{def:SF}
Suppose that $G(V,E)$ satisfies Condition S. Then for any $F\subset V$ with $|F|\leq f$, we denote by $S_F$ the unique source component in graph $G_F$. Note that $|S_F|\geq f+1$ by Condition S. \\For convenience, we will often say that ``$S_F$ is the source component {\it corresponding to} set $F$''.
\end{definition}
\end{framed}
In the graph in Figure \ref{fig:1-reach}, corresponding to $F=\{6\}$, the unique source component $S_F$ in $G_F$ includes the nodes in $\{4,5\}$ (thus, $|S_F|=2$). This graph satisfies Condition S for $f=1$.

\subsection{Implications of Condition S}
\label{s:sync:impli}

Assuming that graph $G(V,E)$ satisfies Condition S, the four lemmas below prove some implications of Condition S. The first three lemmas in this section are used in proving the correctness of our synchronous consensus algorithm (presented in Section \ref{s:sync:algo}). These lemmas are proved in Appendices \ref{a_paths}, \ref{a_out_neighbors}, \ref{a_match_size}, and \ref{a_component_size}, respectively.

\begin{lemma}
\label{l_sf_reach_Fstar}
Suppose that $G(V,E)$ satisfies Condition S.
Given a set $F\subset V$ with $|F|\leq f$, from each node in the corresponding source component $S_F$ to each node in $V$, there exists a directed path that does not include any nodes in $F$ as internal nodes.
\end{lemma}

Note that, with message authentication, a single path from a non-faulty sender to a destination node, which does not include any faulty internal nodes, suffices to deliver messages reliably from the non-faulty sender to the destination node. Thus,
Lemma \ref{l_sf_reach_Fstar} allows us to infer that, {\bf if} set $F$ happens to include all the faulty nodes in a given execution of our algorithm, any information from the source component $S_F$ can be disseminated reliably to all the nodes in the system.

In the next two lemmas, we  consider the special case when $|F|=f$, and $f+1\leq |S_F|\leq 2f$. Lemma \ref{l_out_neighbors} is useful in proving Lemma \ref{l_match_size}.\footnote{An alternate proof of Lemma \ref{l_match_size} without using Lemma \ref{l_out_neighbors} also exists. However, the result in Lemma \ref{l_out_neighbors} may be of independent interest.}

\begin{lemma}
\label{l_out_neighbors}
Suppose that $G(V,E)$ satisfies Condition S.
Suppose that for some set $F\subset V$ with $|F|=f$, the corresponding source component $S_F$ contains $f+p$ nodes, where $0< p\leq f$. Then, for $p \leq k\leq |F|= f$,
any size $k$ subset of set $F$ must have at least $k+1$ outgoing neighbors in $S_F$.
\end{lemma}

For Lemma \ref{l_match_size} below, recall the definitions of $E_{P,Q}$ and bipartite graph $B(P,Q,L)$ from Section \ref{s_terminology}. Lemma \ref{l_match_size} facilitates a novel procedure in the ``state update'' step in our algorithm below.

\begin{lemma}
\label{l_match_size}
Suppose that $G(V,E)$ satisfies Condition S.
Suppose that for some set $F\subset V$ with $|F|=f$, the corresponding source component $S_F$ contains $f+p$ nodes, where $0<p\leq f$. Consider the bipartite graph $B(F,S_F,E_{F,S_F})$. Let $F_1$ be any size $k$ subset of set $F$, where $p\leq k\leq |F|=f$. There exists a matching of size at least $k-p+1$ consisting of edges from $F_1$ to $S_F$.
\end{lemma}

According to Condition S, for any set $F\subset V$, $|F|\leq f$, the source component $S_F$ is of size at least $f+1$. From this we can deduce the following claim.

\begin{lemma}
\label{l_component_size}
Consider $F\subset V$ such that $|F|\leq f$. Then $|S_F|\geq 2f+1-|F|$.
\end{lemma}

\subsection{Synchronous Byzantine Consensus Algorithm}
\label{s:sync:algo}

Our goal here is only to prove that Condition S is sufficient for synchronous Byzantine consensus with message authentication. {\bf No attempt is made here to develop an efficient algorithm.} The proposed algorithm (Algorithm \ref{a_sync_algo}) is presented next. The subsequent discussion below will elaborate on the steps listed in Algorithm \ref{a_sync_algo}.

\begin{algorithm}[tbhp]
\caption{Synchronous Byzantine Consensus with Message Authentication}
\label{a_sync_algo}
\begin{algorithmic}[1]
 \State 
 % \textit{State variable initialization}:
 {\it Initialization}: Each node $v\in V$ initializes state variable $s_v$ to equal its input $x_v$, i.e., $s_v=x_v$

\vspace*{4pt}

\For{each $F\subset V$ such that $|F|=f$}
        \begin{list}{}{}
            \item[(i)] {\it Flooding of state variable values}: Let us denote by $I_{F}$  the set of all the incoming neighbors of $S_F$ in $F$ (thus, $I_F\subseteq F$). The goal in step (i) is to use a ``flooding'' procedure to deliver state variable values from the nodes in $I_{F}\cup S_F$ to all the non-faulty nodes. Section \ref{ss_flood} describes the flooding procedure used in this step.

            \item[(ii)] {\it Updating state variables}:  Each node $v\in V$ performs the procedure in Section \ref{ss_sf} to  update its state variable $s_v$.
            \end{list}   
\EndFor

\vspace*{4pt}

\State\textit{Output step}: $y_v = s_v$
\end{algorithmic}
\end{algorithm}

The algorithm assumes that each node in the system knows graph $G(V,E)$. At the start of the algorithm, each node $v\in V$ has an input from $\{0,1\}$, denoted as $x_v$. Each node $v$ maintains a binary state variable $s_v$ that is initialized in Step 1 of Algorithm \ref{a_sync_algo} to equal $x_v$. The value of the state variable $s_v$ at the end of the algorithm becomes node $v$'s output, as shown in Step 3 of Algorithm \ref{a_sync_algo}. Recall that the output of node $v$ is denoted as $y_v$.

In Step 2, Algorithm \ref{a_sync_algo} iterates through all sets $F\subset V$ such that $|F|=f$, and performs Steps (i) and (ii) in each such iteration (as elaborated in Sections \ref{ss_flood} and \ref{ss_sf}). In step (i), the term ``flooding'' is borrowed from prior literature on networking algorithms that disseminate information across a given network (e.g., \cite{DBLP:journals/corr/cs-NI-0311013,Yi-Gerla-Kwon-flooding-2003}) -- however, our flooding procedure differs from the flooding schemes used in the prior literature, as elaborated soon in Section \ref{ss_flood}.

We assume that digital signatures are used for the purpose of message authentication. We will denote by $M^u$ a message $M$ signed by node $u$. When using digital signatures for authentication, node $u$ would send message $M$ (often called cleartext) and a cryptographically generated signature $Sig(M,Key_u)$ both together, where $Key_u$ denotes node $u$'s private key.
For brevity, as a shorthand for $[M,Sig(M,Key_u)]$, we will represent this message simply as $M^u$. An important property we utilize is that any tampering of a  message signed by a non-faulty node can be detected by the message recipient.

\subsubsection{\bf Step 2(i) of Algorithm \ref{a_sync_algo}: Flooding of State Variable Values}
\label{ss_flood}

We denote by $I_{F}$ the set of all incoming neighbors of $S_F$ in $F$, i.e., $I_F=\{x~|~(x,y)\in E,~x\in F,~y\in S_F\}$. Thus, $I_F\subseteq F$, and each node in $I_F$ has a link to at least one node in $S_F$. Informally, the flooding procedure has two parts. First, the nodes in $I_{F}$ send their state variable values to their outgoing neighbors in $S_F$. Subsequently, each node in $S_F$ is responsible for disseminating the values received from $I_F$, {\bf along with} its own state variable value, to the other nodes.\footnote{We can simplify this step somewhat by requiring the value of a node in $I_F$ to be sent to its outgoing neighbor in $S_F$ only if the corresponding edge is included in the maximal matching identified by Lemma \ref{l_match_size}. For brevity, we do not further elaborate on this optimization.}
The exact method for forwarding the values is discussed below. We now define two sets used in the flooding procedure.
\begin{itemize}
    \item Set $X_v$: Each node $v\in S_F$ maintains a set $X_v$, which is initialized as empty at the start of the flooding procedure in step 2(i) in each iteration of  Algorithm \ref{a_sync_algo}.\footnote{Recall that each iteration of the algorithm corresponds to a different set $F$.} As elaborated below, set $X_v$ is used by node $v$ to collect values received from incoming neighbors in $I_F$. Each value stored in $X_v$ is signed by its sender, and has the form $(u,s)^u$, where $u$ is an incoming neighbor of $v$ in $I_F$, and $s\in\{0,1\}$.

    \item  Set $Y_v$: Each node $v\in V$ maintains a set $Y_v$ of tuples of the form $(u,s)$, where $u\in I_F\cup S_F$ and $s\in\{0,1\}$. Set $Y_v$ is initialized as empty at the start of the flooding procedure in step 2(i). Each tuple in $Y_v$ contains a value from some node in $I_F\cup S_F$, and the corresponding node's identifier.

\end{itemize}

The flooding procedure is performed in rounds. In round 0, the following steps are performed: 
\begin{itemize}

\item In round 0, each node $u\in I_{F}$ sends $(u,s_u)^u$ to each of its outgoing neighbors in $S_F$. Note that $(u,s_u)^u$ represents node $u$'s state variable value $s_u$ (along with the node identifier $u$) signed by node $u$. A faulty node may not follow the algorithm correctly. For instance, a faulty node may not send a value at all, or send an arbitrary value (in fact, a faulty node may also send different values to different neighbors).

\item In round 0, when a node $v\in S_F$ receives a message of the form $(u,s)^u$, with $s\in \{0,1\}$, from an incoming neighbor $u\in I_{F}$:

\hspace*{0.5in} node $v$ adds $(u,s)^u$ to the set $X_v$, and \\
\hspace*{0.5in}  adds $(u,s)$ to the set $Y_v$.

Note that a message of the form $(u,s)^r$, where $u\neq r$, would be ignored in this  step, since it is not signed by node $u$ whose identifier is included in the message. In our algorithm, in general, any message that appears tampered is discarded.
\end{itemize}

In round 0, as described above, set $X_v$ at node $v\in S_F$ collects the signed values from $v$'s incoming neighbors in $I_{F}$. If a certain node $v\in S_F$ has no incoming neighbors in $I_{F}$, then $X_v$ remains empty.

In round 1 of the flooding procedure, the following two steps are performed:
\begin{itemize}
\item Each node $v\in S_F$ adds $(v,s_v)$ to set $Y_v$.
\item Each node $v\in S_F$ sends the message $(v,s_v,X_v)^v$ to all its outgoing neighbors. Note that node $v$ signs the message $(v,s_v,X_v)$ as a whole (i.e., a single signature for the entire message). This property will be useful in proving the correctness of the algorithm.
\end{itemize}

As described soon, the messages sent by the nodes in $S_F$ in round 1 are forwarded by the nodes in $V-F$ during rounds 2 through $|V|$. Thus, each message received by any node in these rounds should have the form $(u,s,X)^u$ where $u\in S_F$, and the set $X$ possibly contains values from $u$'s incoming neighbors in $I_{F}$. As an example, the message may be $(u,1,X)^u$, where $X=\{(p,0)^p, (q,1)^q\}$, $u\in S_F$, and $p,q$ are $u$'s incoming neighbors in $I_{F}$. The procedure in rounds 2 through $|V|$ will use the following notion of a ``correct'' message.

\noindent{\bf ``Correct'' messages:}
Each node $v$, on receiving a message in rounds 1 to $|V|$, performs the following checks. The goal is to discard tampered messages, or messages that do not conform to the algorithm. Any message that passes the checks below is said to be ``correct''; other messages are discarded.
\begin{itemize}
\item Node $v$ checks if the message is of the form $(u,s,X)^u$, where $u\in S_F$ and $s\in\{0,1\}$. That is, the message must be signed by some node $u\in S_F$ with contents of the form $(u,s,X)$ and $s$ in $\{0,1\}$. $u$ may not necessarily be the node from which node $v$ directly received the message.
\item If the above check passes, then node $v$ checks the content of $X$ in the message. $X$ should contain only values in the form $(w,s)^w$, where $w$ is an incoming neighbor of $u$ in $I_{F}$, and $s\in\{0,1\}$. If this condition is also satisfied, then the message $(u,s,X)^u$ has passed the checks.
\end{itemize}

In round $r$ of the flooding procedure in step 2(i), $2\leq r\leq |V|$, the following steps are performed:
\begin{itemize}
\item If node $v\in V$ received a {\bf correct} message of the form $(u,s,X)^u$ in round $(r-1)$ from some incoming neighbor in $V-F$, then node $v$ performs the following steps in round $r$ (note that $u$ may not necessarily be the incoming neighbor from whom $v$ received this message.):
\begin{itemize}
 
    \item Node $v$ adds $(u,s)$ as well as all the values in set $X$ to set $Y_v$. For example, if the received message is $(u,1,X)^u$, where $X=\{(p,0)^p, (q,1)^q\}$, then node $v$ updates $Y_v$ as $Y_v = Y_v \cup \{(u,1)\}\cup \{(p,0),(q,1)\}$.
    \item If $v\in V-F$, then node $v$ forwards the message $(u,s,X)^u$ to all its outgoing neighbors. Note that this message retains the signature by node $u$.
\end{itemize}

In the first step above, only the messages received from the incoming neighbors in $V-F$ are considered (i.e., any messages from the incoming neighbors in $F$ are ignored). Therefore, the second step above needs to be performed only if $v\in V-F$. 
\end{itemize} 
The flooding procedure in step 2(i) is complete at the end of round $|V|$. Set $Y_v$ at node $v$ at the end of the flooding procedure in step 2(i) is used in step 2(ii) to update the state variable $s_v$ at node $v$.

\subsubsection{\bf Step 2(ii) of Algorithm \ref{a_sync_algo}: Updating State Variables}
\label{ss_sf}

In step 2(ii), each node $v\in V$ performs the following operations:
\begin{itemize}
    \item Node $v$ checks whether set $Y_v$ contains, corresponding to every node $w\in S_F$, a tuple of the form $(w,s)$, where $s\in\{0,1\}$. If this  is not true, then node $v$ skips the remaining operations in step 2(ii), and $v$'s state variable $s_v$ is unchanged in the current iteration. If the above condition holds, then $Y_v$ contains values from $\geq |S_F|$ nodes, and $v$ proceeds to the next step. % and node $v$ performs the steps below.
    \item For any node $w\in I_{F}\cup S_F$, if $Y_v$ contains $(w,0)$ and $(w,1)$ both, then node $v$ removes $(w,1)$ from set $Y_v$.\footnote{In this case, node $w$ must be faulty. However, we do not make explicit use of this knowledge in the steps described here. An alternate procedure that uses the knowledge can be easily developed.}.
    \item Using the updated set $Y_v$, node $v$ updates its state variable $s_v$ as described next.
\end{itemize}
Observe that the above procedure ensures that, at any fault-free node $v$, the set $Y_v$ can only include values for nodes in $I_F\cup S_F$ (i.e., values of nodes  {\bf outside} $I_F\cup S_F$ are {\bf not} added to set $Y_v$).
Depending on the number of distinct nodes whose value is included in set $Y_v$, node $v$ uses a different procedure to update its state variable, as described in Cases 1 and 2 below.
\begin{itemize}
\item {\bf Case 1: Set $Y_v$ contains values from at least $2f+1$ distinct nodes in $I_{F}\cup S_F$:}

In this case, node $v$ sets $s_v$ equal to the value in $Y_v$ from the largest number of nodes; a tie between values 0 and 1 is broken in favor of value 0. \vspace*{4pt} \\
{\it \bf An observation to be used to prove the correctness of Algorithm \ref{a_sync_algo}}: In Case 1, node $v$ has collected binary values from at least $2f+1$ nodes. At most $f$ of these nodes may be faulty; thus, a majority of these nodes must be non-faulty. This observation helps us prove the validity property, as discussed later in the proof of correctness. 

\item {\bf Case 2: Set $Y_v$ contains values from at most $2f$ distinct nodes in $I_{F}\cup S_F$:}
Recall that this step may be performed only if set $Y_v$ passed the check at the start of step 2(ii), as described above. Therefore, we know that $Y_v$ contains values from all the nodes in $S_F$. Therefore, in Case 2, we can conclude that $|S_F|\leq 2f$. Also, by Condition S, $|S_F|\geq f+1$. Let us define $p$, $0<p\leq f$, such that $|S_F|=f+p$. Now, we define the following:
\begin{itemize}

\item  $Z_v = \{ u~|~\mbox{$Y_v$ contains a value from node $u\in F$}\}$. Since values only from the nodes $I_F\cup S_F$ are included in $Y_v$, and $F\cap S_F=\emptyset$, $Z_v\subseteq I_F$.
\item  $F_{1v} = F - Z_v$
\item  Let $k=|F_{1v}|$. Thus, $|Z_v|=f-k$. Recall that $|F| = f$, as per Algorithm \ref{a_sync_algo}.
\end{itemize}
$F_{1v}$ contains those $k$ nodes from $F$ whose value is not included in $Y_v$. On the other hand, $Z_v=F-F_{1v}$ contains the $f-k$ nodes in $F$ from which node $v$ has received a value.

Now, $Y_v$ contains values from all the $f+p$ nodes in $S_F$, and by the definition of $k$ above, $Y_v$ also contains values from the $f-k$ nodes in $Z_v$, for a total of $2f+p-k$ nodes.
In case 2, by assumption, this number is no more than $2f$. That is, $2f+p-k\leq 2f$, which implies that $p\leq k$. Also, $k\leq |F|=f$.
 
Thus, in this case, we can apply Lemma~\ref{l_match_size}, and conclude that there exists a matching of size at least $k-p+1$ consisting of edges from $F_{1v}$ to $S_F$. So, let us consider a matching of size {\bf equal to} $k-p+1$ from $F_{1v}$ to $S_F$. Let $M_v$ denote the set of size $2(k-p+1)$ that contains $k-p+1$ nodes from $F_{1v}$ and $k-p+1$ nodes from $S_F$ such that the nodes in $M_v\cap F_{1v}$ are matched with the nodes in $M_v\cap S_F$ in the above matching.

Now consider any edge $(q,r)$ in this matching, where $q\in F_{1v}$ and $r\in S_F$. We show that at least one of these two nodes $q$ and $r$ must be faulty. The proof is by contradiction. Suppose that both $q$ and $r$ are non-faulty. Then in round 0 of flooding in step (i) in the given iteration of Algorithm \ref{a_sync_algo}, node $q$ should have sent $(q,s_q)^q$ to node $r$. Subsequently, in round 1, node $r$ should have sent the message $(r,s_r,X_r)^r$, where $X_r$  contained $(q,s_q)^q$. By the check performed at the start of step (ii), we know that $Y_v$ includes a value from node $r\in S_F$. Thus, node $v$ must have received the above message from node $r$, which also includes a value from node $q$ in set $X_r$. Then, it is not possible for node $q$ to be included in set $F_{1v}$, which is a contradiction. Thus, at least one of $q$ and $r$ must be faulty. Extending this argument to the entire matching above, we can conclude that at least $k-p+1$ nodes among the $2(k-p+1)$ nodes in $M_v$ must be faulty. 

{\bf Updating $s_v$ in Case 2}: From set $Y_v$, node $v$ discards the values from all the nodes in $M_v\cap S_F$. There are $k-p+1$ such nodes, as per the definition of $M_v$. Now node $v$ is left with the values from $(f+p) - (k-p+1) = f+2p-k-1$ nodes in $S_F$. Also, values were received by node $v$ from the $f-k$ nodes in $F-F_{1v}=Z_v$, and no values were received from the nodes in $F_{1v}$ (by the definition of $F_{1v}$). Thus, node $v$ now has values from $(f+2p-k-1)+(f-k)=2f+2p-2k-1$ nodes. Node $v$ sets $s_v$ equal to the value that is received from the largest number of nodes among these $2f+2p-2k-1$ nodes. (In this case, there should not be a tie between 0 and 1 because $2f+2p-2k-1$ is an odd number.)
\\

{\it\bf A useful observation:} We observed above that at least $k-p+1$ nodes among the $2(k-p+1)$ nodes in $M_v$ are faulty. As per the procedure in the previous paragraph, values from none of these nodes was used in computing the new value of $s_v$. Thus, of the $2f+2p-2k-1$ nodes whose value was used above, at most $f-(k-p+1)=f+p-k-1$ are faulty. Finally, observe that $2f+2p-2k-1 = 2(f+p-k-1)+1$. In other words, a majority of the $2f+2p-2k-1$ nodes whose values are used to update $s_v$ are non-faulty.
\end{itemize}

\subsubsection{\bf Correctness of Algorithm \ref{a_sync_algo}:}

Appendix \ref{a:sync:algo:proof} proves correctness of Algorithm \ref{a_sync_algo}. The  observations made in Cases 1 and 2 above ensure that a majority of the values used by node $v$ when updating its state variable $s_v$ come from non-faulty nodes. This helps prove the validity property, which subsequently helps prove the agreement property. As seen above, Lemma \ref{l_match_size} plays an important role in Case 2, and allows us to show that Condition S is sufficient for synchronous consensus.

\section{Asynchronous $\epsilon$-Approximate Byzantine Consensus with Message Authentication}
\label{s:async}

We now present an alternate version of the necessary and sufficient conditions for $\epsilon$-approximate consensus in asynchronous systems,\footnote{Although in Section \ref{s:intro}, we assumed that the inputs are discrete values (i.e., 0 or 1), our results for {\it asynchronous systems} also apply to the case when the inputs are real-valued scalars in a finite range, say, $[0,1]$.}
as stated in Theorem \ref{t_reach_async}. Equivalence of the conditions in Theorem \ref{t_authByzAsync} below and Theorem \ref{t_reach_async} is proved in Appendix \ref{a:async_equiv}. Recall the definition of graph $G_F$ from Section \ref{s:main}.

\begin{theorem}
\label{t_authByzAsync}
In an asynchronous system with message authentication, $\epsilon$-approximate Byzantine consensus in the presence of up to $f$ Byzantine faulty nodes, with any $\epsilon>0$, is possible if and only if graph $G(V,E)$ satisfies both the conditions below.
\begin{itemize}
\item For any $F \subset V$ such that $|F| \leq f$, there is a unique source component $S_F$ in graph $G_F$, and the source component contains at least $2f+1$ nodes. 
    \item For any $F \subset V$ and $F' \subset V$ such that $|F|\leq f$ and  $|F'| \leq f$, the unique source components in $G_F$ and $G_{F'}$ have at least $f+1$ nodes in common.
    That is, $|S_F \cap S_{F'}| \geq f+1$.
\end{itemize}
\end{theorem}
The proof of necessity of the graph conditions above is presented in Appendix \ref{a:async:necessity}. To prove their sufficiency, in Appendix  \ref{a:async:suff}, we present an asynchronous algorithm and prove its correctness. The second condition in Theorem \ref{t_authByzAsync} implies that $S_F \cap S_{F'}$ contains at least one non-faulty node. This observation is crucial in proving that $\epsilon$-approximate agreement is eventually achieved.

The correctness of Theorem \ref{t_reach_async} follows from the correctness of Theorem \ref{t_authByzAsync}, and the equivalence  between the two graph conditions, as proved in Appendix \ref{a:async_equiv}.

\section{Summary}
\label{s:conclu}

Under message authentication, this paper develops tight graph conditions for synchronous and asynchronous consensus in directed graphs in the presence of up to $f$ Byzantine faulty nodes. Further work is needed to develop efficient algorithms for these problems in directed graphs.

% \bibliographystyle{plain}
% \bibliography{paperlist,paperlist2,Tseng,new-bib}

\newpage

\appendix
\noindent {\bf APPENDICES}

\setlength{\parsep}{5pt}

\section{Equivalence of Graph Conditions for Synchronous Consensus}
\label{a:sync_equiv}

In this section, we prove Lemmas \ref{l:sync:1} and \ref{l:sync:2} to show the equivalence of the two graph conditions in Theorems \ref{t_reach_sync} and \ref{t_authByzSync} for synchronous consensus.

Recall the definition of graph $G_F$ from Section \ref{s:main}.

\begin{lemma}
\label{l:sync:1}
Suppose that graph $G(V,E)$ satisfies the 1-reach condition with parameter $\rho=f+1$. Then it satisfies Condition S.
\end{lemma}
\begin{proof}
Suppose that graph $G(V,E)$ satisfies the 1-reach condition with $\rho=f+1$. Consider any set $F\subset V$, $|F|\leq f$. 
\begin{itemize}
\item We will first show that $G_F$ contains a unique source component. The proof is by contradiction. Suppose that $G_F$ contains at least two source components, namely, $S_1$ and $S_2$. Consider $u_1\in S_1$ and $u_2\in S_2$.

By definition, $reach_{u_1}(F)$ does not contain any nodes from $F$. 

Also, $reach_{u_1}(F)$ cannot contain any node from $V-F-S_1$. To show this, observe the following: For a given node $w\in V-F-S_1$ to be included in $reach_{u_1}(F)$, there must be at least one node in $V-F-S_1$ (possibly distinct from $w$) that has an outgoing link to some node in $S_1$; but this is not possible by the definition of $S_1$. 

The above two observations imply that $reach_{u_1}(F)\subseteq S_1$. By a similar reasoning, $reach_{u_2}(F)\subseteq S_2$.
%
% \nitin{This is actually an equality since $S_2$ is strongly connected. I suppose we don't need the equality in the proof here.}
%
Thus, 
$$reach_{u_1}(F)\cap reach_{u_2}(F)\subseteq S_1\cap S_2$$
Since $S_1$ and $S_2$ are two distinct strongly connected components in $G_F$, $S_1\cap S_2=\emptyset$. This implies that 
$$|reach_{u_1}(F)\cap reach_{u_2}(F)|=0$$

The above violates 1-reach condition with $\rho=f+1$, which is a contradiction. Thus, we have proved that there must be a unique source component in $G_F$.
\item Denoting the unique source component in $G_F$ as $S_F$, we now show that $S_F$ must be of size at least $f+1$. By definition, $S_F\subseteq V-F$. Consider any node $v\in S_F$. By an argument similar to that used in the previous item, we know that
$$reach_v(F)\subseteq S_F$$
% \lewis{Should this be $u \neq v$?} 
Suppose that $u=v$. Then the above implies that
$reach_u(F)\cap reach_v(F)\subseteq S_F$,
which, in turn, implies that
$$|reach_u(F)\cap reach_v(F)|\leq |S_F|$$
Note that the 1-reach condition is applicable when $u=v$.
By the 1-reach condition with parameter $\rho=f+1$, we know that $$|reach_u(F)\cap reach_v(F)|\geq f+1$$
The above two inequalities together imply that $|S_F|\geq f+1$.
\end{itemize}
Thus, graph $G(V,E)$ satisfies Condition S.\\
\end{proof}

\begin{lemma}
\label{l:sync:2}
Suppose that graph $G(V,E)$ satisfies Condition S. Then it satisfies the 1-reach condition with parameter $\rho=f+1$. 
\end{lemma}
\begin{proof}
Suppose that graph $G(V,E)$ satisfies Condition S. Consider any set $F\subset V$, $|F|\leq f$. By Condition S, the unique source component $S_F$ corresponding to $F$ is of size at least $f+1$.

Consider nodes $u,v\in V-F$. ($u$ and $v$ may either be distinct or identical nodes.)

Lemma \ref{l_sf_reach_Fstar} holds since Condition S is satisfied. This lemma directly implies that, to node $u\in V-F$, there is a path in $G_F$ from every node $w\in S_F$. Therefore, $S_F\subseteq reach_u(F)$. Similarly, for node $v\in V-F$, $S_F\subseteq reach_v(F)$. Thus, $S_F\subseteq reach_u(F)\cap reach_v(F)$. This, in turn, implies that $$|reach_u(F)\cap reach_v(F)|\geq |S_F|\geq f+1$$ since $|S_F|\geq f+1$ by Condition S. Thus, the 1-reach condition with $\rho=f+1$ is satisfied. 
\end{proof}

\section{``Condition A'' for Asynchronous Consensus}
\label{a:cond_A}

For the convenience of presentation in  the appendices that follow, we now define a ``Condition A'', which consists of the conditions in Theorem \ref{t_authByzAsync}. (Here the letter A in "Condition A" stands for Asynchronous.)

Giving a name to the graph conditions in Theorem 
\ref{t_authByzAsync} makes it easier to refer to them concisely in the proofs presented later in the appendix.

\begin{framed}
\begin{definition}[Condition A]
\label{def:cond-authByzAsync}
Graph $G(V,E)$ is said to satisfy Condition A if it satisfies the two conditions below:

\begin{itemize}
    \item For any $F \subset V$ such that $|F| \leq f$, \\
    \hspace*{0.25in}(i) there is a unique source component $S_F$ in graph $G_F$, and\\
    \hspace*{0.25in}(ii) the size of the unique source component $S_F$ is at least $2f+1$.

    \item For any $F \subset V$ and $F' \subset V$ such that $|F|\leq f$ and  $|F'| \leq f$, the unique source components in $G_F$ and $G_{F'}$ have at least $f+1$ nodes in common.
    That is, $|S_F \cap S_{F'}| \geq f+1$.
\end{itemize}
\end{definition}
\end{framed}

Using the above\footnote{Note that the lower bound of $2f+1$ on the size of a source component is also implied by the lower bound $f+1$ on the size of the intersection of two source components. We include the $2f+1$ lower bound on the source component size separately in Condition A for clarity.} Condition A, we can re-state Theorem \ref{t_authByzAsync} as follows:

\begin{theorem}
\label{t_authByzAsync-restate}
In an asynchronous system with message authentication, $\epsilon$-approximate Byzantine consensus in the presence of up to $f$ Byzantine faulty nodes, with any $\epsilon>0$, is possible if and only if graph $G(V,E)$ satisfies Condition A. 
\end{theorem}

\section{Equivalence of Graph Conditions for Asynchronous Consensus}
\label{a:async_equiv}

In this section, we prove Lemmas \ref{l:async:1} and \ref{l:async:2} to show the equivalence of the graph conditions in Theorems \ref{t_reach_async} and \ref{t_authByzAsync} for asynchronous consensus. Recall that we defined Condition A in Appendix \ref{a:cond_A}, which is a re-statement of the graph conditions in Theorem \ref{t_authByzAsync}.

Recall the definition of graph $G_F$ from Section \ref{s:main}.

\begin{lemma}
\label{l:async:1}
Suppose that graph $G(V,E)$ satisfies the 2-reach condition with parameter $\rho=f+1$. Then it satisfies Condition A.

%\lewis{Should we mention Condition A?}
\end{lemma}
\begin{proof}
Suppose that graph $G(V,E)$ satisfies the 2-reach condition with $\rho=f+1$. 
\begin{itemize}
\item Consider any set $F\subset V$, $|F|\leq f$.  We will show that $G_F$ contains a unique source component. 
The proof of this is very similar to the proof of the uniqueness of the source component in the proof of Lemma \ref{l:sync:1}. In particular, in the last part of that proof, we showed that $|reach_{u_1}(F)\cap reach_{u_2}(F)|=0$. Defining $F_1=F_2=F$, that implies $|reach_{u_1}(F_1)\cap reach_{u_2}(F_2)|=0$, which violates the 2-reach condition with $\rho=f+1$.

In the rest of this proof, as proved above, we will assume that the source component is unique. Corresponding to set $F$, the unique source component is named $S_F$.

\item Now we will show that for any given $F\subset V$, $|F|\leq f$, size of the corresponding unique source component $S_F$ is at least $2f+1$. The proof is by contradiction.

Suppose that for some set $F\subset V$, $|F|\leq f$, the source component size $|S_F|$ is at most $2f$. Also, because $n>f$, the source component is non-empty. Consider any node $u\in S_F$. Then, by the argument in the proof of Lemma \ref{l:sync:1}, $reach_u(F)\subseteq S_F$.

Let us choose $F'\subseteq S_F$ such that $|F'|=\left\lceil \frac{|S_F|}{2}\right\rceil$. Since $0<|S_F|\leq 2f$, it follows that $0<|F'|\leq f$, and $|S_F-F'|\leq |F'|\leq f$. Consider any node $v\in V-F'$. By the definition of the reach sets, $reach_v(F')$ does not contain any nodes in $F'$. Therefore, $reach_v(F')\subseteq V-F'$, and 
\begin{eqnarray}
    reach_u(F)\cap reach_v(F') & \subseteq &  S_F \cap (V-F') \\
    \Rightarrow reach_u(F)\cap reach_v(F') & \subseteq &  (F'\cup (S_F-F')) \cap (V-F') \\
    \Rightarrow reach_u(F)\cap reach_v(F') & \subseteq &  S_F-F' \\
    \Rightarrow |reach_u(F)\cap reach_v(F')| & \leq &  |S_F-F'| ~~\leq~f\\
\end{eqnarray}
This contradicts the 2-reach condition with $\rho=f+1$. Hence, source component $S_F$ must contain at least $2f+1$ nodes.

\item Now we will prove the following: For any $F \subset V$ and $F' \subset V$ such that $|F|\leq f$ and  $|F'| \leq f$, $|S_F \cap S_{F'}| \geq f+1$.

% This part uses some of the ideas used in the second part of the proof of Lemma \ref{l:sync:1}.

 Consider any node $u\in S_F$. Then, as also argued in the proof of Lemma \ref{l:sync:1},
$reach_u(F)\subseteq S_F$. Similarly, consider node $v\in S_{F'}$. Then $reach_v(F')\subseteq S_{F'}$. Thus, $reach_u(F)\cap reach_v(F')\subseteq S_F\cap S_{F'}$. Thus,
$$|reach_u(F)\cap reach_v(F')|\leq |S_F\cap S_{F'}|$$
Also, by the 2-reach condition with $\rho=f+1$, we know that
$$|reach_u(F)\cap reach_v(F')|\geq f+1$$
The above two inequalities together imply that $|S_F\cap S_{F'}|\geq f+1$.

\end{itemize}
Thus, graph $G(V,E)$ satisfies Condition A (and hence it also satisfies the graph conditions in Theorem \ref{t_authByzAsync}).

% \nitin{The above does NOT prove 2f+1 size for $S_F$. That is fine *ONLY IF* we remove that condition from Condition A. Otherwise, we need to add that too here.}

\end{proof}

\begin{lemma}
\label{l:async:2}
Suppose that graph $G(V,E)$ satisfies Condition A. Then it satisfies the 2-reach condition with parameter $\rho=f+1$. 
\end{lemma}
\begin{proof}
Suppose that graph $G(V,E)$ satisfies Condition A. Consider any sets $F_1,F_2\subset V$, $|F_1|,|F_2|\leq f$, and any node $u\in V-F_1$ and any node $v\in V-F_2$. Then, by the conditions in Theorem \ref{t_authByzAsync}, there are unique source components $S_{F_1}$ and $S_{F_2}$ in graphs $G_{F_1}$ and $G_{F_2}$, respectively, each source component containing at least $2f+1$ nodes, and
$$|S_{F_1}\cap S_{F_2}|\geq f+1$$

Now observe that any graph satisfying the Condition A also satisfies Condition S. This implies that Lemma \ref{l_sf_reach_Fstar} holds. This lemma directly implies that, for node $u\in V-F_1$ there is a path in $G_{F_1}$ from every node $w\in S_{F_1}$. Therefore, $S_{F_1}\subseteq reach_u(F_1)$. Similarly, for node $v\in V-F_2$, $S_{F_2}\subseteq reach_v(F_2)$.

Thus, $S_{F_1}\cap S_{F_2}\subseteq reach_u(F_1)\cap reach_v(F_2)$. This, in turn, implies that $$|reach_u(F_1)\cap reach_v(F_2)|\geq |S_{F_1}\cap S_{F_2}|\geq f+1$$ Thus, the 2-reach condition with $\rho=f+1$ is satisfied. 
\end{proof}

\section{Proof of Lemma \ref{l_sf_reach_Fstar}}
\label{a_paths}

We first introduce and prove Lemma \ref{l_sf_reach} below. Lemma \ref{l_sf_reach} will be used to prove Lemma \ref{l_sf_reach_Fstar}.

\begin{lemma}
\label{l_sf_reach}
Suppose that graph $G(V,E)$ satisfies Condition S.
Given a set $F\subset V$ with $|F|\leq f$, from each node in the corresponding source component $S_F$ to each node in $V-F$, there exists a directed path that does not include any nodes in $F$.
\end{lemma}

\begin{proof}
The proof is by contradiction. Suppose that graph $G$ satisfies Condition S, but the claim in the lemma is not true. Then there exist nodes $u\in S_F$ and $v\in V-F$ such that $u$ does not have a path to $v$ that excludes the nodes in $F$. In other words, $u$ does not have a path to $v$ in $G_F$. This has two implications:
\begin{itemize}
    \item $S_F$ is a strongly connected component in $G_F$, and $u\in S_F$ does not have a path to $v$ in $G_F$. Therefore, no node in $S_F$ has a path to $v$ in $G_F$.
    \item Since $S_F$ is a strongly connected component in $G_F$, the above observation implies that $v\not\in S_F$. Thus, $v\in V-F-S_F$.
\end{itemize}

Define $R$ to be the subset of $V-F$ that contains all the nodes in $V-F$ to which the nodes in $S_F$ have a path in $G_F$. Observe that $S_F\subseteq R$. Also, by the definition of $R$, node $v$ above is not in set $R$. Thus, $v\in V-F-R$, and set $V-F-R$ is thus non-empty.

If any node in $R$ were to have a link to some node $w\in V-F-R$, then $w$ would be reachable in $G_F$ from the nodes in $S_F$, and thus $w$ must be in $R$, leading to a contradiction. Thus, there are no links from the nodes in $R$ to any node in $V-F-R$, i.e., $E_{R,(V-F-R)}=\emptyset$.

This, in turn, implies that $V-F-R$ contains a strongly connected component  of $G_F$, say $T$, such that the nodes in $T$ have no incoming neighbors in $V-F-T$. Thus, $T$ is a source component of $G_F$ by definition. Also, $S_F\subseteq R$ and $T\subseteq V-F-R$. Thus, $T$ is distinct from $S_F$, which violates the uniqueness of the source component in $G_F$, as stipulated in Condition S, leading to a contradiction.

This completes the proof.
\end{proof}

% ++++++++++++++++++++++++++

We now use the above lemma to prove Lemma \ref{l_sf_reach_Fstar}. We state Lemma \ref{l_sf_reach_Fstar} here again.

LEMMA \ref{l_sf_reach_Fstar}.
{\it Suppose that graph $G(V,E)$ satisfies Condition S. Given a set $F\subset V$ with $|F|\leq f$, from each node in the corresponding source component $S_F$ to each node in $V$, there exists a directed path that does not include any nodes in $F$ as internal nodes.}

\begin{proof}

If $F=\emptyset$, then the above lemma immediately follows from Lemma \ref{l_sf_reach}.

Now assume that $F$ is non-empty. By Lemma \ref{l_sf_reach}, we know that each node in $S_F$ has a path to each node in $V-F$ that does not include any node in $F$ (thus, no node in $F$ is an internal node on these paths).

We now show the following: to each node in $F$, there exists a link from at least one node in $V-F$.

The proof is by contradiction. 
\begin{itemize}
    \item Suppose that there is no link to some node $w\in F$ from any node in $V-F$.
    
    Let us define $F_0 = F-\{w\}$ Thus, $V-F_0=(V-F)\cup {w}$.

    \item Since there is no link to $w$ from any node in $V-F$, there is no path in $G_{F_0}$ from any node in $V-F$ to $w$.

    \item Consider $S_{F_0}$, the source component in $G_{F_0}$ corresponding to $F_0$. Now, $w\in V-F_0$. Then by Lemma \ref{l_sf_reach}, each node in $S_{F_0}$ must have a path to node $w$ that includes only the nodes from $V-F_0=(V-F)\cup \{w\}$.
    \item The previous two observations imply that $S_{F_0}\subseteq \{w\}$. 
    
    \item However, by condition S, $|S_{F_0}|$ must be at least $f+1$, and $f+1\geq 2$ (because $f>0$). This  contradicts with $S_{F_0}\subseteq \{w\}$.
    \end{itemize}
    This contradiction derived above proves that at least one node, say node $u\in V-F$ has a link to node $w\in F$. That is, $(u,w)\in E$.

    Now, by Lemma \ref{l_sf_reach}, every node in $S_F$ has a path to the node $u$ above that does not include any nodes in $F$. This, and the existence of the link $(u,w)$ implies that there is a path to $w$ from each node in $S_F$ that does not include any node in $F$ as an internal node. This proves the lemma.
\end{proof}

% +++++++++++++++++++++++++++++++++++++++++++

\section{Proof of Lemma \ref{l_out_neighbors}}
\label{a_out_neighbors}

We state the lemma here again.

LEMMA \ref{l_out_neighbors}.
{\it Suppose that graph $G(V,E)$ satisfies Condition S.
Suppose that for some set $F\subset V$ with $|F|=f$, the corresponding source component $S_F$ contains $f+p$ nodes, where $0< p\leq f$. Then, for $p \leq k\leq |F|= f$,
any size $k$ subset of set $F$ must have at least $k+1$ outgoing neighbors in $S_F$.
}

\begin{proof}
The proof is by contradiction.
Suppose that there exists $F_1 \subseteq F$, where $|F_1|=k$ and $p\leq k\leq f$, such that the nodes in $F_1$ have at most $k$ outgoing neighbors in $S_F$ (i.e., the nodes in $F_1$ collectively have outgoing links to at most $k$ nodes in $S_F$).

Consider a set $S_1 \subset S_F$ such that $|S_1|=k$ and $S_1$ contains all the outgoing neighbors of set $F_1$ in $S_F$. If the number of outgoing neighbors of $F_1$ in $S_F$ is strictly less than $k$, then, to  ensure that $|S_1|=k$, set $S_1$ would contain some nodes from $S_F$ that are not outgoing neighbors of $F_1$. This is possible, because by assumption, $S_F$ contains at least $f+1$ nodes, and $k\leq f$. 

Since $|S_F|=f+p$, $|S_1|=k$, and $0<p\leq k \leq f$, it follows that $S_F-S_1$ is non-empty. 

Let us define $F^+ = (F - F_1)\cup S_1$. Observe that $|F^+| = (f - k) + k= f$. $S_{F^+}$ denotes the unique source component corresponding to $F^+$.

Figure \ref{fig:l_out_neighbors} illustrates some of these sets. In the figure, $F_1$ is the white colored subset of $F$, and $S_1$ is the shaded subset of $S_F$. $F^+$ is the union of the shaded subsets of $F$ and $S_F$, i.e., $F^+=(F-F_1)\cup S_1$.

 Now, $V-F^+=(V-F-S_F)\cup F_1\cup (S_F-S_1)$, and $(V-F-S_F) \cup F_1=(V-F^+)-(S_F-S_1)$. The first equality implies that $F^+$ and $S_F-S_1$ are disjoint.

\begin{figure}
    \centering
    \includegraphics[width=0.8\linewidth]{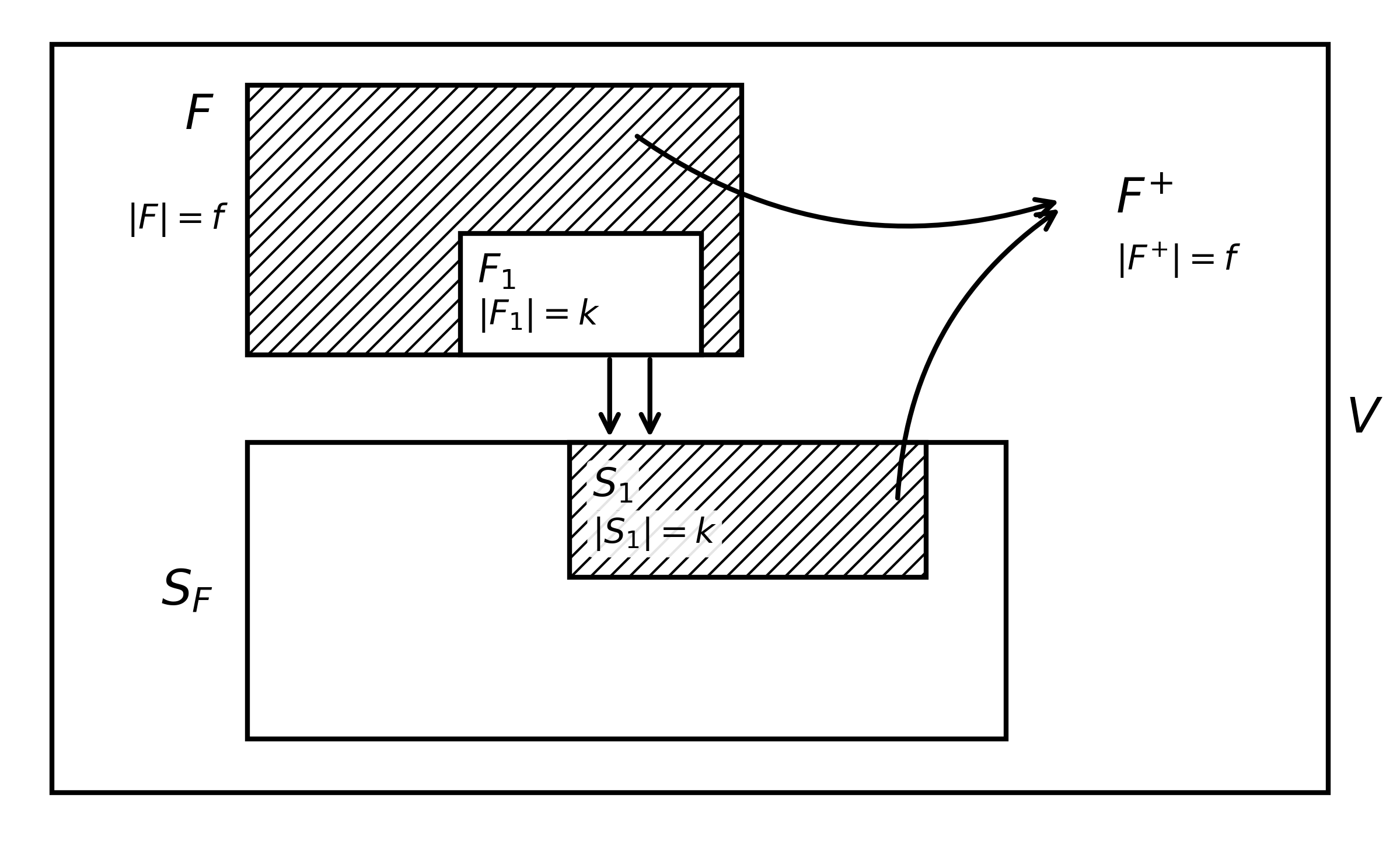}
    \caption{Illustration of sets used in the proof of Lemma \ref{l_out_neighbors}. $F_1$ is the white subset of $F$. $S_1$ is the shaded subset of $S_F$. $F^+$ is the union of the shaded subsets of $F$ and $S_F$, i.e., $F^+=(F-F_1)\cup S_1$.}
    \label{fig:l_out_neighbors}
\end{figure}

Now, by the definition of source component $S_F$, there are no links from $V-F-S_F$ to $S_F-S_1$. Also, since $S_1$ contains all the outgoing neighbors of $F_1$ in $S_F$, there are no links from $F_1$ to $S_F-S_1$. As noted earlier, $(V-F-S_F) \cup F_1=(V-F^+)-(S_F-S_1)$. These three observations together imply that there are no links from the nodes in $(V-F^+)-(S_F-S_1)$ to the nodes in $S_F-S_1$.
This, together with the definition of the source component, implies that $(S_F - S_1)$ must contain a source component in $V-F^+$. Recall that Condition S requires a unique source component. Consequently, $S_{F^+}\subseteq S_F-S_1$.
(Recall that $S_F-S_1$ is non-empty, and $F^+$ and $S_F-S_1$ are disjoint.) 

Therefore, $|S_{F^+}|\leq |S_F-S_1|= (f+p)-k = f+(p-k) \leq f$, since $p\leq k$. However, this contradicts Condition S, which requires each source component to be of size at least $f+1$.

This concludes the proof.
\end{proof}

\section{Proof of Lemma \ref{l_match_size}}
\label{a_match_size}

We state the lemma here again.

LEMMA \ref{l_match_size}.
{\it Suppose that graph $G(V,E)$ satisfies Condition S.
Suppose that for some set $F\subset V$ with $|F|=f$, the corresponding source component $S_F$ contains $f+p$ nodes, where $0<p\leq f$. Consider the bipartite graph $B(F,S_F,E_{F,S_F})$. Let $F_1$ be any size $k$ subset of set $F$, where $p\leq k\leq |F|=f$. There exists a matching of size at least $k-p+1$ consisting of edges from $F_1$ to $S_F$.
}
\begin{proof}

Note that in the bipartite graph $B$ in the lemma, only the edges from $F$ to $S_F$ are included.

Let us define set $N$ consisting of $p-1$ {\bf new} nodes -- thus, the nodes in $N$ do not belong to $V$. Let us also create a set $L$ of {\bf new} edges that contains an edge from each node in $F$ to each node in $N$. Now, we define a new bipartite graph $B^*(F,S_F\cup N, E_{F,S_F}\cup L)$. Effectively, $B^*$ is obtained by augmenting $B$ with the nodes in $N$ and the edges in $L$. Figure \ref{fig:l_match_size} illustrates $B$ and $B^*$.

\begin{figure}
    \centering
    \includegraphics[width=0.8\linewidth]{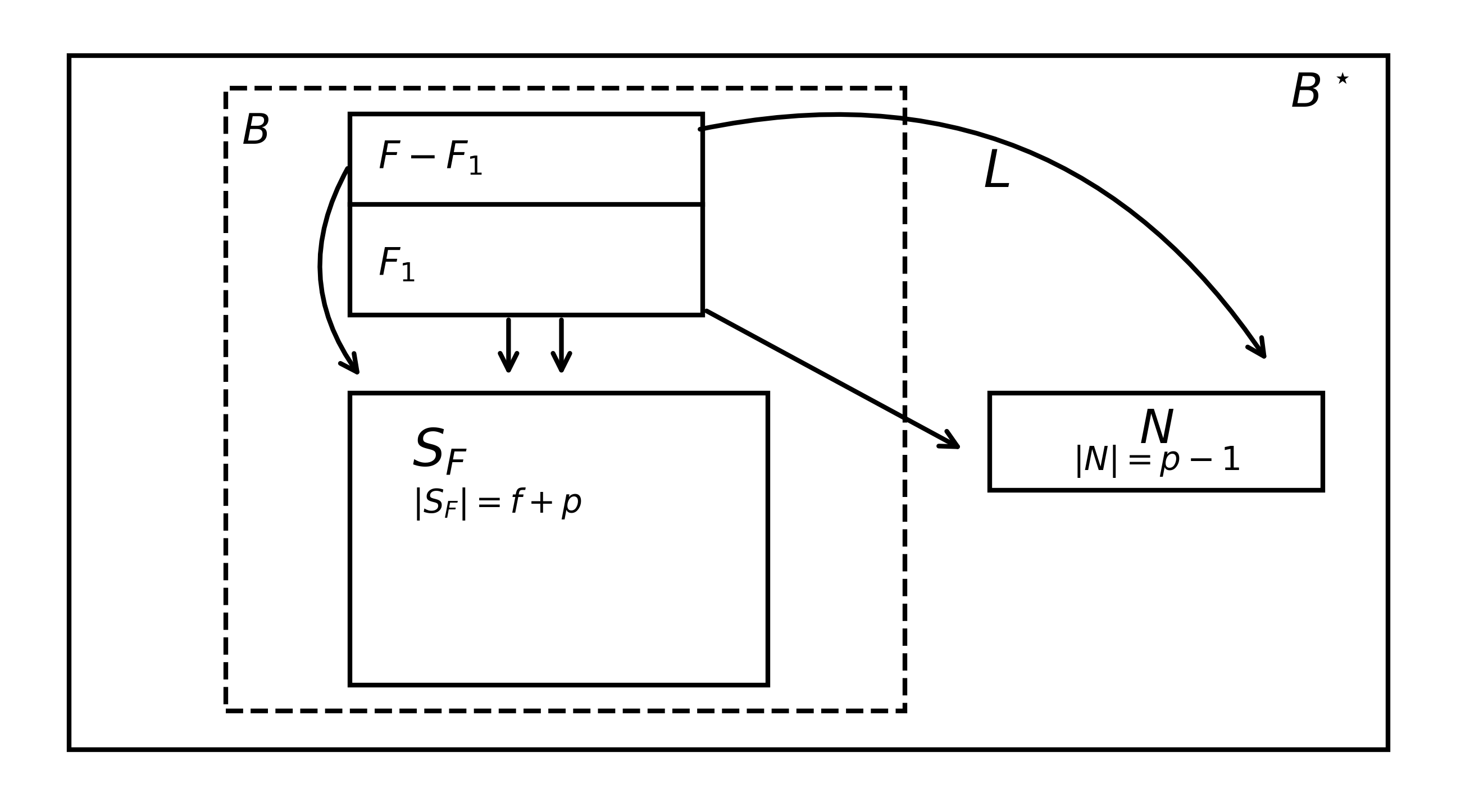}
    \caption{Illustration of sets used in the proof of Lemma \ref{l_match_size}.}
    \label{fig:l_match_size}
\end{figure}

Consider a non-empty subset $R$ of $F$. Let size $|R|$ be denoted as $r$. There are two cases, depending on the value of $r$, recalling that $|F|=f$.
Now we consider each of the two cases:
\begin{itemize}
\item Case 1: $p\leq r\leq |F|=f$: Due to Lemma \ref{l_out_neighbors}, $R$ has at least $r+1$ outgoing neighbors in bipartite graph $B$, and hence also in $B^*$.
\item Case 2: $1\leq r\leq p-1$: Due to the augmentation procedure above (to obtain $B^*$ from $B$), nodes in $R$ have all the nodes in $N$ as the outgoing neighbors. Thus, $R$ has at least $p-1$ outgoing neighbors in bipartite graph $B^*$. Since, in case 2, $r\leq p-1$, this implies that $R$ in this case has at least $r$ outgoing neighbors in bipartite graph $B^*$.
\end{itemize}
The above two observations together imply that every $R$, where $R\subseteq F$, has at least $|R|$ outgoing neighbors in $B^*$.
Thus, we can apply Hall's Theorem to bipartite graph $B^*$, which implies that there is a matching of size $|F|=f$ from $F$ to $S_F\cup N$. Of the $f$ links in this matching, at most $p-1$ links may be to the new nodes in $N$ (because $|N|=p-1$). So, there is a matching of size $\geq f-(p-1)=f-p+1$ from $F$ to $S_F$.

Now consider any $F_1\subseteq F$ of size $k$, $p\leq k\leq |F|=f$. Observe that $F_1$ is obtained by removing $|F-F_1|=f-k$ nodes from $F$. Then, since there is a matching of size at least $f-p+1$ from $F$ to $S_F$, there must be a matching of size at least $(f-p+1)-(f-k)=k-p+1$ from $F_1$ to $S_F$. 

This proves the lemma.
\end{proof}

\section{Proof of Lemma \ref{l_component_size}}
\label{a_component_size}
We first re-state Lemma \ref{l_component_size}.

LEMMA \ref{l_component_size}
{\it Suppose that graph $G(V,E)$ satisfies Condition S. Consider $F\subset V$ such that $|F|\leq f$. Then $|S_F|\geq 2f+1-|F|$.}
\begin{proof}
By Condition S, $|S_F|\geq f+1$. Therefore, the claim in the lemma is true when $|F|=f$.

For $|F|<f$, the proof is by contradiction. Suppose that the claim is not true for some set $F\subset V$ with $|F|<f$. Thus, $|S_F|< 2f+1 - |F|$.

Now consider a set $F'\subset V$ of size equal to $f$ such that $F \subset F'$ and $F'-F\subset S_F$. Such a set $F'$ can surely be formed because $|S_F|\geq f+1$. Let $S_{F'}$ be the source component for $F'$.

Now, by the definition of a source component, the nodes in $S_F$ do not have any incoming links from the the nodes in $V-F-S_F$. Since $F\subset F'$, it follows that the nodes in $S_F$ do not have any incoming links from the nodes in $V-F'-S_F$. Consequently, the source component $S_{F'}$ corresponding to $F'$ should be contained within $S_F$. That is, $$S_{F'}\subseteq S_F$$
Now, by the definition of a source component, $S_{F'}$ does not contain any nodes from $F'$. Combining the above two observations, we obtain that
 $$S_{F'} \subseteq S_F - (F'-F)$$
The above observation along with the fact that $F'-F\subset S_F$ and $F\subset F'$, implies that
\begin{eqnarray*}
|S_{F'}| & \leq & |S_F|-(|F'|-|F|)\\
\Rightarrow |S_{F'}| & < & (2f+1 - |F|) - (f-|F|) ~~\mbox{~~because $|S_F|<2f+1 - |F|$ and $|F'|=f$}\\
\Rightarrow |S_{F'}| & < & f+1
\end{eqnarray*}

This contradicts  Condition S. This completes the proof of the lemma.
\end{proof}

\section{Correctness of the Synchronous Consensus Algorithm (Algorithm \ref{a_sync_algo})}
\label{a:sync:algo:proof}

Assume that graph $G(V,E)$ satisfies Condition S.

Let $F^*$ be the actual set of Byzantine faulty nodes in a given execution of the synchronous consensus algorithm, Algorithm \ref{a_sync_algo}. In this section, we prove that the execution satisfies the termination, validity, and agreement properties specified in Section \ref{s:intro} for consensus in synchronous systems.

\subsection*{Proof of the Termination Property}

Termination in bounded time occurs because the algorithm performs $n\choose f$ iterations (for different $F\subset V$, each of size $f$), with each such iteration taking a bounded amount of time.

\subsection*{Proof of the Validity Property}

To show validity, we first prove Lemma \ref{l_validity} below. Each iteration of Algorithm \ref{a_sync_algo} uses a different set $F$. Lemma \ref{l_validity} below applies to each such iteration. The claim in the lemma only applies to the non-faulty nodes, that is, the nodes in $V-F^*$.

\begin{lemma}
\label{l_validity}
At the end of any given iteration of Algorithm \ref{a_sync_algo}, the state variable $s_v$ of any node $v\in V-F^*$ is equal to the state variable of some node in $V-F^*$ at the start of that iteration.
\end{lemma}
\begin{proof}
We consider two possibilities for any node $v\in V-F^*$:
\begin{itemize}
\item Node $v$ does not update $s_v$ in the current iteration: Then the claim is trivially true, since the value of $s_v$ at the end of the iteration equals its own  value at the start of the iteration.

\item Node $v$ updates its state variable in step (ii): From the observations made in Cases 1 and 2 both in Section \ref{ss_sf}, we know that a majority of the nodes whose binary values are used to update $s_v$ in Step (ii) are non-faulty. Since only a minority of these nodes are faulty, the new value of $s_v$ is guaranteed to be equal to the value received from a non-faulty node which, in turn, equals that node's state variable value at the start of the iteration.
\end{itemize}
This concludes the proof of Lemma \ref{l_validity}.
\end{proof}

Now, we will show validity inductively. At the start of the first iteration in step 2 of Algorithm \ref{a_sync_algo}, the state variable value of each non-faulty node equals its own input. Now suppose that at the start of the $i$-th iteration, $i\geq 1$, in step 2 of Algorithm \ref{a_sync_algo}, the state variable value of each non-faulty node equals the input of a non-faulty node. Then, Lemma \ref{l_validity} implies that the state of each non-faulty node at the end of the $i$-th iteration, and hence at the start of the $(i+1)$-iteration, equals the input of a non-faulty node. By induction, this implies that the state variable of a non-faulty node after the last iteration of the algorithm, and hence the non-faulty node's output, equals the input of a non-faulty node. Thus, the algorithm achieves the validity property.

\noindent{\bf Proof of the Agreement Property}

$F^*$ denotes the actual set of faulty nodes in a given execution of the algorithm. In this section, let us consider the earliest iteration in which $F^*\subseteq F$. We will refer to this iteration as the ``{\bf deciding iteration}''. In the deciding iteration, all the nodes in $V-F$ are non-faulty, because $F^*\subseteq F$. The nodes in $F-F^*$ are also non-faulty. Also, all the nodes in $S_F$ are non-faulty, since $S_F\subseteq V-F$.

We now show that for any two non-faulty nodes $u,v\in V-F^*$, at the end of step 2(i) in the deciding iteration, $Y_u=Y_v$.
\begin{itemize}
\item Suppose that, for $w\in S_F$, value $s$ of node $w$ is included in $Y_u$ (i.e., $(w,s)\in Y_u)$ at the end of step 2(i) in the deciding iteration. Due to the correctness checks applied during flooding to detect message tampering, this implies that node $w$ must have sent value $s$ in round 1 of flooding. This message is forwarded by the non-faulty nodes in $V-F$ in round 2 through round $|V|$. By Lemma \ref{l_sf_reach_Fstar}, there exists a path from $w$ to every node in $V$ (and hence to every non-faulty node in $V-F^*$) that does not contain any node in $F$ as an internal node. Any such loop-free path must be of length at most $|V|-1$. Thus, the faulty nodes in $F^*\subseteq F$ cannot prevent node $w$'s message from reaching the nodes in $V-F^*$ by the end of round $|V|$ of flooding. Therefore, the message sent by node $w$ in round 1 will be received by all the nodes in $V-F^*$ by the end of round $|V|$ of flooding. Thus, $(w,s)$ will be included by node $v$ in set $Y_v$ by the end of step 2(i) in the deciding iteration.
\item Suppose that, for $w\in I_{F}$, and value $s$ of node $w$ is included in $Y_u$ (i.e., $(w,s)\in Y_u$) at the end of step 2(i) in the deciding iteration. Then there must exist a node $q\in S_F$ that included $(w,s)^w$ in its set $X_q$ in round 1 of flooding in the deciding iteration. Now, due to the check performed at the start of step (ii), we know that node $v$ must have received $(q,s_q,X_q)^q$ by the end of step (i), and thus, $v$ has received $(w,s)^w$, and $(w,s)\in Y_v$ at the end of step 2(i) in the deciding iteration.

\item Finally, observe that $Y_u$ and $Y_v$ can only contain values for the nodes in $I_{F}\cup S_F$.
\end{itemize}
The above three observations prove that $Y_u=Y_v$ at the end of step 2(i) in the deciding iteration. Therefore, following the deterministic procedure in step 2(ii), nodes $u$ and $v$ will both update their state variables to an identical value. That is, $s_u=s_v$ at the end of step 2(ii) in the deciding iteration, for any $u,v\in V-F^*$. This, in conjunction with Lemma \ref{l_validity} being applied to all the subsequent iterations of Algorithm \ref{a_sync_algo}, implies that the outputs of the non-faulty nodes satisfy the agreement property. \\

This proves the correctness of Algorithm \ref{a_sync_algo} for synchronous Byzantine consensus with message authentication. In other words, this proves that Condition S is sufficient for synchronous consensus. Since we have also proved (in Appendix \ref{a:sync_equiv}) that Condition S is equivalent to the reach condition in Theorem \ref{t_reach_sync}, the above also proves that the condition in Theorem \ref{t_reach_sync} is sufficient for synchronous consensus.

\section{Necessity of Condition S in Synchronous Systems}
\label{a:sync:necessity}

Recall that ``Condition S'' was defined in Definition \ref{def:cond-authByzSync}. 
Our goal in this appendix is to prove the necessity of Condition S for synchronous consensus. Necessity of Condition S would also imply the necessity of the condition in Theorem \ref{t_reach_sync}, due to the equivalence proved in Appendix \ref{a:sync_equiv}.

This proof is by contradiction. Suppose that a given graph $G(V,E)$ does \textit{not} satisfy Condition S, and there exists a synchronous algorithm $\mathcal A$ that achieves exact Byzantine consensus on graph $G(V,E)$ under message authentication in the presence of up to $f$ Byzantine faults.

The proof of necessity is in two parts, depending on which part of Condition S is violated by graph $G(V,E)$.

\noindent\textbf{Part 1: Uniqueness of the source component}

Suppose that there exists $F\subseteq V$ with $|F|\leq f$ such that $G_F$ contains at least two distinct
source components $S_0$ and $S_1$.

Recall that each source component $S_i$ ($i=1,2$) is a strongly connected component in $G_F$ that has no incoming links from any node in $V-F-S_i$.
Consider the following three executions of Algorithm $\mathcal A$, in each of which, all the nodes in $F$ are Byzantine faulty, and simply crash before the start of the execution (thus, the nodes in $F$ do not send any messages in any of these three executions):

\begin{itemize}
  \item Execution $\mathsf{E}^{(0)}$: Nodes in $F$ crash at the start of the execution. Each non-faulty node in $V-F$ has input $0$.
  By the validity condition, since each non-faulty node's input is $0$, all the nodes in $S_0\subseteq V-F$ output $0$. Suppose that the last of the nodes in $S_0$ decides its output by time $\tau_0$.

  \item Execution $\mathsf{E}^{(1)}$: Nodes in $F$ crash at the start of the execution. Each non-faulty node in $V-F$ has input equal to $1$. By the validity condition, since each non-faulty node's input is 1, the nodes in $S_1\subseteq V-F$ output $1$. Suppose that the last of the nodes in $S_1$ decides its output by time $\tau_1$.

  \item Execution $\mathsf{E}$: Nodes in $F$ crash at the start of the execution. The nodes in $S_0$ have input $0$, the nodes in $S_1$ have input $1$, and all the other non-faulty nodes have input 0.
\end{itemize}

Now, consider execution $E$. By definition, $S_0$ has no incoming neighbors in $V-F-S_0$, and nodes in $F$ send no messages;
thus the nodes in $S_0$ have the same local histories in $\mathsf{E}$ and $\mathsf{E}^{(0)}$ up to time $\tau_0$,
so they will choose output $0$ in execution $\mathsf{E}$ too (as they do in $\mathsf{E}^{(0)}$). Similarly, nodes in $S_1$ choose output 1 in $\mathsf{E}$.
Hence, the agreement condition for synchronous consensus is violated, resulting in  a contradiction. This proves that there must be a unique source component corresponding to any $F\subset V$, $|F|\leq f$.

\noindent\textbf{Part 2: Size of the unique source component $S_F$  must be at least $f+1$}

Part 1 above proved that graph $G_F$ must contain a unique source component corresponding to set $F$, which we refer to as $S_F$. In the proof in Part 2, we assume that the source component is unique.

We now prove by contradiction that the size of the unique source component for any given $F\subset V$, $|F|\leq f$, must be at least $f+1$.
Assume that there exists $F\subset V$ with $|F|\leq f$ such that the size of the unique source component $S_F$ in $G_F$ is at most $f$. Since $n>f$, we know that $|S_F|\geq 1$. Thus, $1\leq |S_F|\leq f$.

We consider two possibilities separately: (i) $V-F-S_F=\emptyset$, and (ii) $V-F-S_F\neq\emptyset$.

Let us first consider the case when $V-F-S_F=\emptyset$. Thus, $V=F\cup S_F$.
In this case, since $|S_F|\leq f$ and $|V|=n>f$, $F$ cannot be empty.
We now show an execution in which consensus is violated. 
In particular, consider an execution $\mathsf{E}$ of Algorithm $\mathcal A$ in which the nodes in $S_F$ have input 0, the nodes in $F$ have input 1, and all the nodes in $V=F\cup S_F$ are non-faulty.
\begin{itemize}
    \item Execution $\mathsf{E}$ is indistinguishable to the nodes in $S_F$ from another execution, say $\mathsf{E}^{(0)}$, in which the nodes in $F$ are faulty, but follow the algorithm correctly with input 1, and the nodes in $S_F$ are non-faulty with input 0. In $\mathsf{E}^{(0)}$, due to the validity condition, the nodes in $S_F$ must choose output 0. Therefore, the nodes in $S_F$ must choose output 0 in execution $\mathsf{E}$ too.
    
    \item Similarly, execution $\mathsf{E}$ is indistinguishable to the nodes in $F$ from another execution, say $\mathsf{E}^{(1)}$, in which the nodes in $S_F$ are faulty, but follow the algorithm correctly with input 0, and the nodes in $F$ are non-faulty with input 1. In $\mathsf{E}^{(1)}$, due to the validity condition, the nodes in $F$ must choose output 1. Therefore, the nodes in $F$ must choose output 1 in execution $\mathsf{E}$ too.

    \item Thus, as noted above in execution $\mathsf{E}$, the nodes in $S_F$ choose output 0, and the nodes in $F$ choose output 1. Since the nodes in $S_F\cup F$ are non-faulty in execution $\mathsf{E}$, the agreement property is violated.  
\end{itemize}

Now, we consider the case when $V-F-S_F\neq\emptyset$. Consider the following executions of  algorithm $\mathcal A$ when $V-F-S_F\neq\emptyset$.

\begin{itemize}
\item Execution $\mathsf{E}^{a}$: The nodes in $S_F$ have input 0, the nodes in $F$ have input 1, and all the other nodes in $V-F-S_F$ have input 1. The nodes in $S_F$ are faulty.
The nodes in $S_F$ follow algorithm $\mathcal A$ but with one caveat: the nodes in $S_F$ pretend that the nodes in $F$ have crashed before the start of the execution (thus, the messages from $F$ are ignored by the nodes in $S_F$, and similarly, any messages that would be sent to $F$ from $S_F$ are simply discarded before actually being sent to $F$).

In this case, due to the validity condition, the nodes in $V-F-S_F$ must choose output 1. Note that $V-F-S_F$ is non-empty.

\item Execution $\mathsf{E}^{b}$: The nodes in $S_F$ have input 0, the nodes in $F$ have input 1, and all the other nodes in $V-F-S_F$ have input 0. The nodes in $F$ are faulty, but follow the algorithm $\mathcal A$ correctly with two  caveats: (i) the nodes in $F$ do not send messages to the nodes in $S_F$, and discard any message received from $S_F$; and (ii) the nodes in $F$ behave to nodes in $V-F-S_F$ as if the  nodes in $S_F$ have discarded any message received from $F$ and have  not sent any message to $F$. 

In this case, due to the validity condition, the nodes in $S_F$ must choose output 0.

\item Execution $\mathsf{E}^{c}$: The nodes in $S_F$ have input 0, the nodes in $F$ have input 1, and all the other nodes in $V-F-S_F$ have input 1. The nodes in $F$ are faulty, but follow the algorithm $\mathcal A$ correctly with two caveats: 
(i) the nodes in $F$ do not send messages to the nodes in $S_F$, and discard any message received from $S_F$; and (ii) the nodes in $F$ behave to nodes in $V-F-S_F$ as if the  nodes in $S_F$ have discarded any message received from $F$ and have  not sent any message to $F$. Note that these two caveats are identical in both $E^b$ and $E^c$; however, the behavior of faulty nodes in $F$ may not be entirely identical in the two executions, because the nodes in $V-F-S_F$ have different inputs in $E^b$ and $E^c$ (and these nodes nodes may send messages to the nodes in $F$).

To the nodes in $V-F-S_F$, execution $\mathsf{E}^{c}$ is indistinguishable from execution $\mathsf{E}^{a}$. 
This is because of the following three reasons: (i) The nodes in $V-F-S_F$ have input $1$ in both executions; (ii) In execution $\mathsf{E}^{c}$, the nodes in $F$ behave to nodes in  $V-F-S_F$ in a way that the messages from $F$ are ignored by $S_F$, and any messages that would be sent to $F$ from $S_F$ are discarded. (iii) In execution $\mathsf{E}^{a}$, nodes in $S_F$ behave to nodes in $V-F-S_F$ as nodes in $F$ have crashed before the start of the execution. (Note that the effect of this  behavior of $S_F$ in $\mathsf{E}^{a}$ is that the messages from $F$ are ignored by the nodes in $S_F$, and any messages that would be sent to $F$ from $S_F$ are simply discarded before actually being sent to $F$.) So, to the nodes in $V-F-S_F$, execution $\mathsf{E}^{c}$ is indistinguishable from execution $\mathsf{E}^{a}$. Since they choose output 1 in execution $\mathsf{E}^{a}$, in execution $\mathsf{E}^{c}$ also the nodes in $V-F-S_F$ must choose output 1.

To the nodes in $S_F$, execution $\mathsf{E}^{c}$ is indistinguishable from execution $\mathsf{E}^{b}$ (recall that there are no links from $V-F-S_F$ to $S_F$, by the definition of a source component). Therefore, in execution $\mathsf{E}^{c}$, the nodes in $S_F$ must choose output 0. 

Thus, in execution $\mathsf{E}^{c}$, the nodes in $S_F$ choose output 0, and the nodes in $V-F-S_F$ choose output 1. Since the nodes in $V-F=(V-F-S_F)\cup S_F$ are non-faulty in execution $\mathsf{E}^{c}$, the above violates the agreement property.
\end{itemize} 

The contradiction above proves that the source component size must be at least $f+1$.

\section{Sufficient Condition for Asynchronous Consensus} 
\label{a:async:suff}

The goal in this section is to prove the sufficiency of the conditions for asynchronous consensus stated in Theorem \ref{t_authByzAsync}. Recall that, for convenience of presentation, we stated these conditions as Condition A in Definition \ref{def:cond-authByzAsync} (in Appendix \ref{a:cond_A}). Thus, we now  show the sufficiency of Condition A to achieve consensus in an asynchronous system with message authentication. 

To prove the sufficiency of Condition A, let us assume that graph $G(V,E)$ satisfies Condition A. We now present an asynchronous Byzantine consensus algorithm, and prove its correctness.

Although in Section \ref{s:intro},
 we assumed that the inputs are discrete values (i.e., 0 or 1), our results for {\it asynchronous systems} also apply to the case when the inputs are real-valued scalars in a finite range, say, $[0,1]$.\\

\noindent{\bf Asynchronous Byzantine Consensus Algorithm}

The algorithm proceeds in asynchronous rounds, with rounds numbered 1, 2, 3, etc.  Each node $v\in V$ maintains a real-valued state variable $s_v$, which is initialized to equal its input $x_v$ at the start of the algorithm.

In each round $r$, each node $v\in V$ also maintains set $X_{v,u}$ for each $u\in V$. The algorithm proceeds in asynchronous rounds. Messages belonging to round $r$ include the round number, so that they can be identified as round $r$ messages. For simplicity, we will not include the round number explicitly in the messages in our description below.

The algorithm description below refers to three procedures: (a) Correctness check for a message, (b) Completeness check for set $X_{v,v}$, and (c) Update procedure for state variable $s_v$. We will present the details of these three procedures {\it after} presenting the rest of the algorithm. \\

\hrule

The algorithm proceeds in asynchronous rounds.
Each node $v\in V$ performs the following two steps, namely steps (i) and (ii), in each asynchronous round $r$, $1\leq r\leq R_{max}$, where  $R_{max}$ is the total number of rounds performed. We will discuss how $R_{max}$ is determined later.
\begin{list}{}{}
 \item[(i)] At the start of round $r$, node $v$ initializes set $X_{v,v}$ to be equal to $\{(v,s_v)^v\}$, and for each $u\in V$ such that $u\neq v$, it initializes $X_{v,u}$ to be an empty set. Then, node $v$ sends a round $r$ message containing $(v,X_{v,v})^v$ to all its outgoing neighbors.
 \item[(ii)] On receipt of a round $r$ message of the form $(u,X)^u$ from an incoming neighbor $u$, node $v$ performs the following steps if the message passes the ``correctness check'' as described below:
 \begin{list}{}{}
    \item[(a)] Update $X_{v,u}$: If $X_{v,u}\subset X$, then $X_{v,u} = X$
     \item[(b)] Update $X_{v,v}$:~~ $X_{v,v} = X_{v,v}\cup X$
     \item[(c)] Node $v$ sends a round-$r$ message $(v,X_{v,v})^v$ to all its outgoing neighbors.
     \item [(d)] If set $X_{v,v}$ passes the {\it completeness check} (as described below) corresponding to some set $F_v$, $|F_v|\leq f$, then node $v$ performs an {\it update procedure} (as described below) to update its state variable $s_v$ using the set $X_{v,v}$ that satisfied the {\it completeness check}, and then begins asynchronous round $r+1$ of the algorithm.
 \end{list}
\end{list} 

When node $v$ has completed $R_{max}$ asynchronous rounds, it sets its output $y_v=s_v$.

In addition to the above steps, each node $v$ performs the following step, independent of which round it is in itself:
\begin{itemize}
\item On receipt of a message (irrespective of its round number) of the form $(u,X)^u$ from an incoming neighbor $u$, if the message passes the ``correctness check'', node $v$ forwards the message $(u,X)^u$ to all its outgoing neighbors. (To prevent the same information from circulating in the network indefinitely, we can add a mechanism that limits the number of times the same information is forwarded 
to $|V|$. We omit the details of such a mechanism here, since it is peripheral to the {\it correctness} of the algorithm, and pertains to its overhead. It should suffice here to note that such a mechanism indeed exists.)  \\
\end{itemize}

\hrule

~

Now we elaborate on three procedures used in the above algorithm.
\begin{itemize}
\item {\bf Correctness check of message $(u,X)^u$}: This check is designed to ensure that the information in the message is correctly signed. Since the message is of the form $(u,X)^u$, we already know that the message was originally sent by node $u$ and bears the signature of node $u$. The ``correctness'' of set $X$ remains to be checked. The message passes the correctness check if $X$ only contains values in the form $(w,s)^w$, where $w\in V$ and $s$ is a real number.\\

(Note that if X contains a value of the form $(p,s)^q$ where $p\neq q$, it will fail the above check. In general, any tampered value in set $X$ will fail the above check.)  \\

\item {\bf Completeness check of $X_{v,v}$} in step ii(c): Set $X_{v,v}$ passes the completeness check if there exists a set $F_v\subset V$, $|F_v|\leq f$, such that the following properties are satisfied:
\begin{itemize}
    \item If there exist multiple (i.e., distinct) values in $X_{v,v}$ from any node $q$, then node $q\in F_v$. Note that in this case, node $q$ must be surely faulty.
    \item For each node $w$ in the source component $S_{F_v}$ corresponding to $F_v$, set $X_{v,w}$ at node $v$ contains exactly one value for each node in $S_{F_v}$.
    \item Set $X_{v,v}$ at node $v$ contains exactly one value for each node in $S_{F_v}$. This condition is actually implied by the above two conditions. However, we include it as a separate condition check for clarity. \\
\end{itemize}

\item {\bf Update procedure for $s_v$}: Recall that the values in sets $X_{v,v}$ and $X_{v,w}$ above are of the form $(x,s)^x$ : In this example, the value $s$ corresponds to node $x$. If node $x$ is non-faulty, then node $x$ must have indeed sent this value in the current round (because faulty nodes cannot tamper messages signed by non-faulty nodes).

From step ii(c) recall that, to perform the update of $s_v$, node $v$ uses the same instance of its set $X_{v,v}$ that satisfied the completeness check (i.e., no more values are added to set $X_{v,v}$ before the update of $s_v$). Note that the set $X_{v,v}$ may potentially contain an identical value received from multiple nodes -- each such instance of a value in $X_{v,v}$ is treated as a distinct value (i.e., we consider the multiset formed by the values in $X_{v,v}$). Secondly, if for any node $x$, there are multiple distinct values in $X_{v,v}$, then node $x$ must be faulty. The following steps are performed to update $s_v$:
\begin{list}{}{}
    \item[(I)]  If there is any node $q$ for which there are multiple (i.e., distinct) values in $X_{v,v}$, then all such values are removed from $X_{v,v}$. In this case,  node $q$ must be faulty. Suppose that values from $\phi$ nodes are removed in this step.
    \item[(II)] The remaining values in $X_{v,v}$ are sorted in an increasing order. If the values from two different nodes are identical, the value from the node with the smaller identifier is considered smaller for the purpose of sorting. Using this sorted order, and $\phi$ from the previous step, node $v$ removes the smallest $f-\phi$ values, and the largest $f-\phi$ values, from among the remaining values in set $X_{v,v}$.
    Set $X_{v,v}$ after removing these values is used in step III.
    \item[(III)] Let $M_v$ and $m_v$ denote the maximum and the minimum values among the remaining values in $X_{v,v}$. The new value of $s_v$ is computed\footnote{It is possible to use other methods for updating $s_v$ in this step. For example, we may compute the new value of $s_v$ as the average of all the remaining values in $X_{v,v}$. The choice of the update method affects the proof of correctness somewhat.} as $\frac{M_v+m_v}{2}$.
    
\end{list}
 
\end{itemize}

\subsection*{Correctness of the Asynchronous Consensus Algorithm}

Assume that graph $G(V,E)$ satisfies Condition A described in Definition \ref{def:cond-authByzAsync}.

Let $F^*$ denote the actual set of faulty nodes in a given execution of the asynchronous algorithm. We will show that the asynchronous algorithm satisfies the termination, validity and $\epsilon$-agreement properties described in Section \ref{s:intro} for asynchronous algorithms.

{\bf Proof of the Termination Property:} We first show that the progress of a non-faulty node through each asynchronous round cannot be impeded by the faulty nodes. In particular, we show that, for each non-faulty node $v$, the completeness check in each round will  be satisfied eventually.
\begin{itemize}
    \item Note that $S_{F^*}$ contains only non-faulty nodes, since $S_{F^*}\subseteq V-F^*$ by definition.
    \item Now observe that Condition A is stronger than Condition S, and therefore, Lemma \ref{l_sf_reach_Fstar} holds under Condition A too. Lemma \ref{l_sf_reach_Fstar} implies that, each node in $S_{F^*}$ has a path to every non-faulty node (including
    the nodes in $S_{F^*}$), such that these paths do not include any faulty nodes as internal nodes.
    \item The above property ensures that  the completeness check can eventually pass at any non-faulty node $v$ for $F_v=F^*$. In general, the completeness check may be satisfied sooner for some $F_v$ different from $F^*$. Therefore, each non-faulty node will eventually complete any given asynchronous round. Thus, it will eventually complete $R_{max}$ rounds, and select an output. \\
\end{itemize}

{\bf Proof of the Validity Property:} Consider any given round, and consider the update procedure used by a non-faulty node $v$ for updating its state variable $s_v$ in that round. Observe that the $\phi$ nodes whose values are dropped in step I of the update procedure are necessarily faulty (because a non-faulty node sends only one state variable value in each round, and a faulty node cannot tamper messages signed by a non-faulty node).

After the completion of step I of the update procedure, each remaining value in $X_{v,v}$ comes from a distinct node. If the $f-\phi$ smallest values that are dropped from $X_{v,v}$ in step II all come from faulty nodes, then the remaining values in $X_{v,v}$ must all come from non-faulty nodes. Similar argument holds if the largest $f-\phi$ values dropped in step II all come from faulty nodes.

On the other hand, in step II, if at least one of the smallest $f-\phi$ values dropped comes from some non-faulty node, and at least one of the largest $f-\phi$ values dropped comes from some non-faulty node, then every value remaining in $X_{v,v}$ is within the range of these two values from non-faulty nodes.

In each case, the values used in computing new $s_v$ in step III are within the range of the values received from non-faulty nodes in that round. Therefore, the new value of $s_v$ is also within the range of the values received from non-faulty nodes in that round.

Now recall that the nodes send their current state variable values at the start of each given round. Also, at the start of round 1,  the state variable value at each node is equal to the node's input. Then applying the above argument inductively, we conclude that the validity property is satisfied by the outputs of the non-faulty nodes. \\
 
{\bf Proof of the $\epsilon$-Approximate Agreement Property:} The agreement property can be proved analogous to the proofs of approximate consensus in some of the prior work (e.g., \cite{DBLP:conf/podc/GhineaLW22,vaidya_PODC12}). The argument below is similar to that used in \cite{DBLP:conf/podc/GhineaLW22}.

Consider any two non-faulty nodes $u,v\in V-F^*$ that have both completed round $r$. Each node updates its state variable, as described in the update procedure. Recall the maximum and minimum values $M_v$ and $m_v$ at node $v$, as defined in the update procedure's step III in round $r$, and analogously, the values $M_u$ and $m_u$ at node $u$ in its round $r$.
\begin{itemize}
\item By graph condition A, $S_{F_u}\cap S_{F_v}$ contains at least $f+1$ nodes, which must include at least one non-faulty node. In the rest of this proof of the {\it $\epsilon$-approximate agreement} property, let us consider one such {\bf non-faulty} node $w \in S_{F_u}\cap S_{F_v}$. In the following, we consider sets $X_{u,w}$ and $X_{v,w}$ at nodes $u$ and $v$, respectively, at the time when their sets $X_{u,u}$ and $X_{v,v}$ pass the completeness check with respect to $F_u$ and $F_v$, respectively, in round $r$.

Due to steps ii(b) and ii(c) of the algorithm, we have $X_{u,w}\subseteq X_{u,u}$ and $X_{v,w}\subseteq X_{v,v}$.

\item $X_{u,w}$ and $X_{v,w}$ each equal some instance of set $X_{w,w}$ sent by node $w$. Since $w$ is non-faulty, we know that the set $X_{w,w}$ at node $w$ grows monotonically. Therefore, without loss of generality, hereafter let us assume that $$X_{v,w}\subseteq X_{u,w}$$

\item By the completeness check, we know that set $X_{v,v}$ that passed the completeness check at node $v$ contains exactly one value for each node in $S_{F_v}$. Note that this set $X_{v,v}$ is then used by node $v$ to update $s_v$. This implies that none of the values from the nodes in $S_{F_v}$ are dropped from $X_{v,v}$ in step I of the update procedure at node $v$ in round $r$. This observation, and the fact that $|S_{F_v}|\geq 2f+1$ by Assumption A, in turn, imply that, after step II, the rank-$(f+1)$ value among the values received by node $v$ from the nodes in $S_{F_v}$ is still included in set $X_{v,v}$ when node $v$ performs step III in round $r$. Suppose that this rank-$(f+1)$ value is $\psi_{f+1}$, and it is from some node $p\in S_{F_v}$. For future reference, note that the value $\psi_{f+1}$ from node $p$ is included in the interval $[m_v,M_v]$ because that value from node $p$ is still included in $X_{v,v}$ when node $v$ performs step III in round $r$.

Also, value $\psi_{f+1}$ for node $p$ must also be included in $X_{v,w}$, since (i) by the completeness check, $X_{v,w}$ must include a value for node $p\in S_{F_v}$, and (ii) because $X_{v,w}\subseteq X_{v,v}$, and by the completeness check, $X_{v,v}$ contains exactly one value for $p\in S_{F_v}$, the value in $X_{v,w}$ for node $p\in S_{F_v}$ cannot be different from the value for node $p$ in $X_{v,v}$.

\item Since $X_{v,w}\subseteq X_{u,w}$, we can infer that value $\psi_{f+1}$ from node $p$ is also included in $X_{u,w}$, and hence in $X_{u,u}$, when node $u$ performs step I of the update procedure in round $r$. However, if node $p$ happens to be faulty, it is possible that $X_{u,u}$ also includes additional values for node $p$. (This is unlike $X_{v,v}$, which includes exactly one value for node $p$, because $p\in S_{F_v}$.) Now, there are two cases possible at node $u$:
\begin{itemize}
\item Value $\psi_{f+1}$ from node $p$ is {\bf not} removed from $X_{u,u}$ in step I of the update procedure in round $r$ at node $u$: In this case, we now show that this value will not be removed from $X_{u,u}$ in step II either.

Recall that $X_{v,w}$ contains exactly one value for each node in $S_{F_v}$, and value $\psi_{f+1}$ from node $p$ has rank $(f+1)$ among these values from $S_{F_v}$. Since $|S_{F_v}|\geq 2f+1$, this implies that $X_{v,w}$ contains values from $2f$ other nodes that are {\bf distinct from node $p$}, $f$ of which are ranked lower than $\psi_{f+1}$ from node $p$, and $f$ of which are ranked higher than $\psi_{f+1}$ from node $p$. Also, when $u$ passes the completeness check at node $u$ for $F_u$, $X_{v,w}\subseteq X_{u,w}\subseteq X_{u,u}$. This implies that the value $\psi_{f+1}$ from node $p$ will not be removed from $X_{u,u}$ in step II of the update procedure at node $u$.

Thus, the value $\psi_{f+1}$ is in $X_{u,u}$ when node $u$ performs step III of the update procedure in round $r$. Thus, $\psi_{f+1}\in [m_u,M_u]$. As observed earlier, we also have  $\psi_{f+1}\in [m_v,M_v]$. Hence the intersection of $[m_u,M_u]$ and $[m_v,M_v]$ is non-empty.

\item Value $\psi_{f+1}$ from node $p$ is removed from $X_{u,u}$ in step I of the state update procedure in round $r$ at node $u$: Thus, value $\psi_{f+1}$ from node $p$ is among the values from $\phi$ nodes that are removed in step I at node $u$ in round $r$. Therefore, $\phi\geq 1$ at node $u$ when it performs the state update in round $r$. Additionally, node $p$ must be faulty.

By the completeness check, as also observed in the previous case, $X_{u,u}$ contains values from $2f$ distinct nodes, which are also {\bf distinct from node $p$}, $f$ of which are ranked lower than $\psi_{f+1}$ from node $p$, and $f$ of which are ranked higher than $\psi_{f+1}$ from node $p$. This implies that, in step III, at least one value ranked lower and at least one value ranked higher than $\psi_{f+1}$ from node $p$ remains in $X_{u,u}$. This, in turn, means that $[m_u,M_u]$ contains $\psi_{f+1}$. Thus, the intersection of the intervals $[m_u,M_u]$ and $[m_v,M_v]$ contains $\psi_{f+1}$, and hence the intersection is non-empty.
\end{itemize}
\end{itemize}
Thus, in both cases, we have shown that the intervals $[m_u,M_u]$ and $[m_v,M_v]$ have a non-empty intersection.

The rest of the proof is relatively straightforward. Let us define $s_{min}(r)$ and $s_{max}(r)$ as the smallest and the largest values among the state variables of non-faulty nodes (i.e., nodes in $V-F^*$) at the start of the $r$-th iteration. Now, the argument presented in proving the {\it validity} property implies that $[m_u,M_u]\subseteq [s_{min}(r),s_{max}(r)]$, and $[m_v,M_v]\subseteq [s_{min}(r),s_{max}(r)]$. 

For clarity in the proof below, let us denote by $s_q(r)$ the value of the state variable at any node $q$ at the start of round $r$. Then, we have the following:
\begin{eqnarray}
    s_u(r+1) = \frac{M_u+m_u}{2} \\
    s_v(r+1) = \frac{M_v+m_v}{2}
\end{eqnarray}

As shown earlier, $[m_u,M_u]\cap [m_v,M_v]\neq\emptyset$. This implies that $M_v\geq m_u$ and $M_u\geq m_v$. We will use this implication soon.

Now, without loss of generality, suppose that $s_v(r+1)\leq s_u(r+1)$. Then, 
\begin{eqnarray*}
|s_u(r+1)-s_v(r+1)|& = & s_u(r+1)-s_v(r+1) \mbox{~~because $s_v(r+1)\leq s_u(r+1)$} \\
& = & \frac{M_u+m_u}{2}-\frac{M_v+m_v}{2} \\
& = & \frac{M_u-m_v}{2}-\frac{M_v-m_u}{2} \\
& \leq & \frac{M_u-m_v}{2} \mbox{\hspace*{0.1in}because  $M_v\geq m_u$, as noted above} \nonumber \\
& \leq & \frac{s_{max}(r)-s_{min}(r)}{2} \\
&& \hspace*{0.1in} \mbox{since $[m_u,M_u]\subseteq [s_{min}(r),s_{max}(r)]$ \& $[m_v,M_v]\subseteq [s_{min}(r),s_{max}(r)]$}
\end{eqnarray*}
Since the above claim is true for any two non-faulty nodes $u,v\in V-F^*$, we get
\[
s_{max}(r+1)-s_{min}(r+1)\leq \frac{s_{max}(r)-s_{min}(r)}{2}
\]
Thus, the range of the state variable values of the non-faulty nodes shrinks at least in half after each asynchronous round. Thus, after a sufficiently large number of rounds (which we call $R_{max}$), the range of values will shrink to a length smaller than $\epsilon$. Note that an upper bound on $R_{max}$ can be computed as a function of $\epsilon$ and the length of the range of the input values of the non-faulty nodes. In particular, if the input values are restricted to the interval $[0,1]$, then we can choose $R_{max} = \left\lceil log_2 \frac{1}{\epsilon}\right\rceil$, for $0<\epsilon< 1$. (When inputs are restricted in $[0,1]$, for $\epsilon\geq 1$, the input values of non-faulty nodes already satisfy $\epsilon$-approximate agreement.)

\section{Necessary Condition for Asynchronous Consensus}
\label{a:async:necessity}

Recall that, for the convenience of presentation, we defined ``Condition A''  (in Definition \ref{def:cond-authByzAsync}) that includes the conditions from Theorem \ref{t_authByzAsync}.
Our goal in this section is to prove the necessity of Condition A for asynchronous consensus.

The proof is by contradiction. Suppose that a given graph $G(V,E)$ does not satisfy Condition A, and there exists an asynchronous algorithm $\mathcal A$ that achieves $\epsilon$-approximate Byzantine consensus on graph $G(V,E)$ under message authentication for any $\epsilon>0$ in the presence of up to $f$ Byzantine faults.

The proof of necessity is in three parts, depending on which part of Condition A is violated by graph $G(V,E)$.

\noindent\textbf{Part 1: Uniqueness of the source component}

Suppose that there exists $F\subseteq V$ with $|F|\leq f$ such that $G_F$ contains at least two distinct
source components $S_0$ and $S_1$.

Recall that each source component $S_i$ ($i=1,2$) is a strongly connected component in $G_F$ that has no incoming links from any node in $V-F-S_i$.
Consider the following three executions of Algorithm $\mathcal A$, in each of which, all the nodes in $F$ are Byzantine faulty, and simply crash before the start of the execution (thus, the nodes in $F$ do not send any messages in any of these three executions):

\begin{itemize}
  \item Execution $\mathsf{E}^{(0)}$: Nodes in $F$ crash at the start of the execution. Each non-faulty node in $V-F$ has input $0$.
  By the validity condition, since each non-faulty node's input is $0$, all the nodes in $S_0\subseteq V-F$ output $0$. Suppose that the last of the nodes in $S_0$ decides its output by time $\tau_0$.

  \item Execution $\mathsf{E}^{(1)}$: Nodes in $F$ crash at the start of the execution. Each non-faulty node in $V-F$ has input $1$. By the validity condition, since each non-faulty node's input is 1, the nodes in $S_1\subseteq V-F$ output $1$. Suppose that the last of the nodes in $S_1$ decides its output by time $\tau_1$.

  \item Execution $\mathsf{E}$: Nodes in $F$ crash at the start of the execution. The nodes in $S_0$ have input $0$, the nodes in $S_1$ have input $1$, and all the other non-faulty nodes have input 0.
\end{itemize}

Now, consider execution $E$. By definition, $S_0$ has no incoming neighbors in $V-F-S_0$, and nodes in $F$ send no messages;
thus the nodes in $S_0$ have the same local histories in $\mathsf{E}$ and $\mathsf{E}^{(0)}$ up to time $\tau_0$,
so they will choose output $0$ in execution $\mathsf{E}$ too (as they do in $\mathsf{E}^{(0)}$). Similarly, nodes in $S_1$ choose output 1 in $\mathsf{E}$.
Hence, for any $\epsilon<1$, the $\epsilon$-approximate agreement condition is violated, resulting in  a contradiction. This proves that there must be a unique source component corresponding to any $F\subset V$, $|F|\leq f$.

\noindent\textbf{Part 2: Size of the unique source component $S_F$  must be at least $2f+1$}

Part 1 above proved that graph $G_F$ must contain a unique source component, which we refer to as $S_F$. In the proof in Parts 2 and 3, we assume that the source component is unique.

We now prove by contradiction that the size of the unique source component for any given $F\subset V$, $|F|\leq f$, must be at least $2f+1$.
Assume that there exists $F\subset V$ with $|F|\leq f$ such that the size of the unique source component $S_F$ in $G_F$ is at most $2f$. Since $n>f$, we know that $|S_F|\geq 1$.
We consider two cases:
\begin{itemize}
\item $|S_F|=1$: Since $f>0$, $|S_F|=1$ implies that $|S_F|\leq f$. We can prove that $|S_F|$ must be at least $f+1$. The proof of this is essentially identical to that in Part 2 of the proof in Appendix \ref{a:sync:necessity}, with the only difference being that, in the asynchronous case, the $\epsilon$-approximate agreement property is violated for $\epsilon<1$, when inputs are in $[0,1]$. (The proof in Appendix \ref{a:sync:necessity} instead addresses the violation of exact agreement property, because that proof is for a synchronous system.) We do not repeat the entire proof here, for the sake of brevity. Thus, we have shown that $|S_F|$ cannot equal 1.

\item $|S_F| \geq 2$: The rest of the proof below assumes that $|S_F|\geq 2$.
\end{itemize}

Partition $S_F$ into two non-empty disjoint sets $L$ and $R$ (i.e., $S_F = L \cup R$  and $L\cap R=\emptyset$), such that $0<|L|\leq f$ and $0<|R|\leq f$. Such a partition exists because $2f\geq |S_F|\geq 2$. Now consider the following three executions of Algorithm $\mathcal A$:

\begin{itemize}
    \item Execution $\mathsf{E}$: The nodes in $V-R$ (including the nodes in $L$) have input 0, the nodes in $R$ have input 1. The nodes in $F$ are faulty, and crash at the start of the execution (i.e., do not send any messages during the execution). In this case, the nodes in $L\cup R$ do not receive any messages from the nodes in $V-L-R=V-S_F$ (recall that the nodes in $V-F-S_F$ do not have links to the nodes in $S_F$).

     This execution is admissible because (i) $|F|\leq f$, (ii) up to $f$ Byzantine faulty nodes are allowed, and (iii) Byzantine faulty nodes can potentially crash at any time.

    Now, the nodes in $L$ and $R$ must decide their output eventually.
 \begin{itemize}
      \item Let $\tau$ denote the time by which the last of the nodes in $L\cup R$ chooses its output.
        \item    To satisfy the agreement condition, all the nodes in $L\cup R$ must choose the same output. Suppose that the nodes in $L\cup R$ choose their output equal to some value $\alpha$.\\
    \end{itemize}

    \item Execution $\mathsf{E}^{(0)}$: Now we construct an execution $\mathsf{E}^{(0)}$ which, up to time $\tau$, will appear to the nodes in $L$ indistinguishable from execution $\mathsf{E}$. Suppose that the nodes in $V-R$ (including $L$) have input 0, and the nodes in $R$ have input 1. The nodes in $R$ are Byzantine faulty, but they \textit{follow the algorithm specification} for algorithm $\mathcal{A}$ with their input being 1. The nodes in $V-R$ are non-faulty with input $0$. The asynchronous scheduler delays every message sent from the nodes in $F$ to the nodes in $L\cup R=S_F$ until after time $\tau$. Thus, until time $\tau$, the nodes in $L\cup R$ will not receive any messages from the nodes in $V-L-R=V-S_F$ (recall that the nodes in $V-F-S_F$ cannot directly send messages to $S_F$).
    
    This execution is admissible because (i) $|R|\leq f$, (ii) up to $f$ Byzantine faulty nodes are allowed, and (iii) the asynchronous model allows arbitrarily long message delays.
  
  Now, to the nodes in $L$, up to time $\tau$, execution $\mathsf{E}^{(0)}$ appears identical to execution $\mathsf{E}$. Therefore, the nodes in $L$ will choose an output equal to $\alpha$ by time $\tau$ (as they do in Execution $\mathsf{E}$).
  
  However, in execution $\mathsf{E}^{(0)}$, all the non-faulty nodes (i.e., the nodes in $V-R$, including $L$) have input 0; therefore, by the validity condition, the nodes in $L$ must choose output 0. This implies that $\alpha$ must be equal to 0. \\

  \item Execution $\mathsf{E}^{(1)}$: Now we construct an execution $\mathsf{E}^{(1)}$ which, up to time $\tau$, will appear to the nodes in $R$ indistinguishable from execution $\mathsf{E}$. Suppose that the nodes in $V-L$ (including $R$) have input 1, and the nodes in $L$ have input 0. The nodes in $L$ are Byzantine faulty, but they \textit{follow the algorithm
  specification} for algorithm $\mathcal{A}$ with their input being 0.
  The asynchronous scheduler delays every message sent from nodes in $F$ to the nodes in $L\cup R=S_F$ until after time $\tau$.
  
  This execution is admissible because (i) $|F|\leq f$, (ii) up to $f$ Byzantine faulty nodes are allowed, and (iii) the asynchronous model allows arbitrarily long message delays.
  
  Now, to the nodes in $R$, up to time $\tau$, execution $\mathsf{E}^{(1)}$ appears identical to execution $\mathsf{E}$. Therefore, the nodes in $R$ will choose output equal to $\alpha$ by time $\tau$ (as they do in Execution $\mathsf{E}$). As determined in the discussion of execution $\mathsf{E}^{(0)}$ above, $\alpha=0$. Thus, the nodes in $R$ choose their output equal to 0 in execution  $\mathsf{E}^{(1)}$ also.
  
  However, in execution $\mathsf{E}^{(1)}$, all the non-faulty nodes (i.e., the nodes in $V-L$, including $R$) have input 1; therefore, by the validity condition, the nodes in $R$ must choose output 1.
  
  This contradicts the conclusion in the previous paragraph that the nodes in $R$ choose output 0.
\end{itemize}
The contradiction derived above proves the necessity of $|S_F|\geq 2f+1$.\\

\noindent\textbf{Part 3: The intersection of two source components must contain at least $f+1$ nodes}

In Parts 1 and 2 above, we have already proved the necessity of two aspects of  Condition A. Therefore, in Part 3, we now assume that the source component $S_F$ corresponding to any $F\subset V$, $|F|\geq f$, is unique and $|S_F|\geq 2f+1$. (Note that if these conditions do not hold then we can derive a contradiction using the proofs in Parts 1 and 2.)

We now prove by contradiction that the intersection of any two source components must contain at least $f+1$ nodes. So, suppose that the graph does not satisfy this condition, and there exist $F,F'\subseteq V$ with $|F|,|F'|\leq f$,  such that $|S_F|,|S_{F'}|\geq 2f+1$, and
$$|S_F\cap S_{F'}|\leq f$$

Let $X=S_F\cap S_{F'}$; thus, $|X|\le f$. Define
$A=S_F- X$ and $B=S_{F'}- X$. Thus, $S_F=A\cup X$ and $S_{F'}=B\cup X$. Since $|S_F|\geq 2f+1$, $|S_{F'}|\geq 2f+1$,  $|X|\leq f$ and $f>0$, both $A$ and $B$
are nonempty. The following observations are used in the proof below:
\begin{itemize}
\item There are no links from the nodes in $V-F-S_F$ to the nodes in $S_F=A\cup X$.

\item $B \cap S_F=\emptyset$ (by the definition of $X$), and $B-F\subseteq V-F-S_F$. Therefore, the previous observation implies that there are no links from the nodes in $B-F$ to the nodes in $S_F=A\cup X$.

\item There are no links from the nodes in $V-F-S_{F'}$ to the nodes in $S_{F'}=B\cup X$.

\item $A\cap S_{F'}=\emptyset$ (by the definition of $X$), and $A-F'\subseteq V-F'-S_{F'}$. Therefore, the previous observation implies that there are no links from the nodes in $A-F'$ to the nodes in $S_{F'}=B\cup X$.
\end{itemize}

Now consider the following three executions of Algorithm $\mathcal A$:
 \begin{itemize}
    \item Execution $\mathsf{E}^{(0)}$: The nodes in $F$ are faulty, and crash at the start of the execution (i.e., do not send any messages during the execution). All the nodes in $V-F$ have input 0.
    
    In this case, the nodes in $S_F=A\cup X$ do not receive any messages from the nodes in $V-S_F=V-A-X$, either directly from those nodes or via any other nodes (Recall that the nodes in $V-F-S_F$ do not have links to the nodes in $S_F$. While the nodes in  $V-F-S_F$ may have links to the nodes in $F$, the nodes in $F$ crash at the start of the execution.) Thus, in execution $\mathsf{E}^{(0)}$, the nodes in $A$ do not receive any authenticated  information that originated at the nodes in $V-F-S_F$. (We will also use this fact in the analysis of execution $\mathsf{E}$ below.)

     This execution is admissible because (i) $|F|\leq f$, (ii) up to $f$ Byzantine faulty nodes are allowed, and (iii) Byzantine faulty nodes can potentially crash at any time.

    Now, the nodes in $A$ must decide their output eventually.
 \begin{itemize}
      \item Let $\tau_0$ denote the time by which the last of the nodes in $A$ chooses its output.
        \item  To satisfy the validity condition, all the nodes in $A$ must choose output 0, because all the non-faulty nodes in $V-F$ have input 0.\\
    \end{itemize}

       \item Execution $\mathsf{E}^{(1)}$: The nodes in $F'$ are faulty, and crash at the start of the execution (i.e., do not send any messages during the execution). All the nodes in $V-F'$ have input 1.
    
    In this case, the nodes in $S_{F'}=B\cup X$ do not receive any messages from the nodes in $V-S_{F'}=V-B-X$, either directly from those nodes or via any other nodes (the reasoning is analogous to that in the previous execution). Thus, in execution $\mathsf{E}^{(1)}$, the nodes in $B$ do not receive any authenticated  information that originated at the nodes in $V-F'-S_{F'}$. (We will also use this fact in the analysis of execution $\mathsf{E}$ below.)
    
    Now, the nodes in $B$ must decide their output eventually.
 \begin{itemize}
      \item Let $\tau_1$ denote the time by which the last of the nodes in $B$ chooses its output.
        \item  To satisfy the validity condition, all the nodes in $B$ must choose output 1, because all the non-faulty nodes in $V-{F'}$ have input 1.\\
    \end{itemize}

    \item Execution $\mathsf{E}$: Now we construct an execution $\mathsf{E}$ such that, (i) up to time $\tau_0$, $\mathsf{E}$ will appear to the nodes in $A$ indistinguishable from execution $\mathsf{E^{(0)}}$, and (ii) up to time $\tau_1$, $\mathsf{E}$ will appear to the nodes in $B$ indistinguishable from execution $\mathsf{E^{(1)}}$.

    In execution $\mathsf{E}$, the nodes in $X$ are Byzantine faulty, and the nodes in $V-X$ are non-faulty. The nodes in $A$ have input 0, and the nodes in $B$ have input 1. The remaining nodes in $V-X-A-B$ may have arbitrary input value (say, 0).

    In execution $\mathsf{E}$, suppose that any messages from the nodes in $F$ to the nodes in $A$ are not delivered to the nodes in $A$ until after time $\tau_0$. Similarly, any messages from the nodes in $F$ to the nodes in $B$ are not delivered to the nodes in $B$ until after time $\tau_1$.
    
    Finally, in execution $\mathsf{E}$, up to time $\tau_0$, suppose that the nodes in $X$ send the same messages to the nodes in $A$ that they send to $A$ in execution $\mathsf{E^{(0)}}$, following the same schedule as in $\mathsf{E^{(0)}}$. Similarly, up to time $\tau_1$, suppose that the nodes in $X$ send the same messages to the nodes in $B$ that they send to $B$ in execution $\mathsf{E^{(1)}}$, following the same schedule as in $\mathsf{E^{(1)}}$. Note that this behavior by set $X$ is possible for two reasons: (i) As noted earlier, there are no links from $B-F$ to $S_F=A\cup X$, and no links from $A-F'$ to $S_{F'}=B\cup X$. (ii) In execution $\mathsf{E}$ and $\mathsf{E^{(0)}}$ both, the nodes in $A$ do not receive any messages from $F$ until time $\tau_0$, and similarly, in executions $\mathsf{E}$ and $\mathsf{E^{(1)}}$ both, the nodes in $B$ do not receive any messages from $F'$ until time $\tau_1$.

    Thus, due to the behavior of the nodes in $X$, and the delayed messages from $F$, to the nodes in $A$, execution $\mathsf{E}$ appears indistinguishable from execution $\mathsf{E^{(0)}}$ up to time $\tau_0$. Thus, as in execution $\mathsf{E^{(0)}}$, the nodes in $A$ choose output 0 by time $\tau_0$.

    Similarly, due to the behavior of the nodes in $X$, and the delayed messages from $F'$, to the nodes in $B$, this execution appears indistinguishable from execution $\mathsf{E^{(1)}}$ up to time $\tau_1$. Thus, as in execution $\mathsf{E^{(1)}}$, the nodes in $B$ choose output 1 by time $\tau_0$.

    Since $A$ and $B$ are both non-empty, the above violates the $\epsilon$-approximate agreement property for $\epsilon<1$. Thus, algorithm $\mathcal A$ is incorrect.
    
    This is a contradiction. Therefore, we have proved that $S_{F}\cap S_{F'}$ must contain at least $f+1$ nodes.
  
\end{itemize}

In parts 1, 2 and 3 together, we have proved the necessity of Condition A for asynchronous consensus.

\section{Total Order on Reach Conditions in Table \ref{t:directed}}
\label{a:reach}

Recall that reach conditions in Definition \ref{def:reach1}. It is easy to prove the order below on the reach conditions. However, we include the proofs here for completeness. The reader may omit Appendix \ref{a:reach} without a loss of continuity.

\begin{tabular}{p{50pt}p{10pt}p{60pt}p{10pt}p{50pt}p{10pt}p{60pt}p{10pt}p{50pt}}
1-reach \newline with $\rho=1$ &
$\Leftarrow$ & 1-reach\newline with $\rho=f+1$ &
$\Leftarrow$ & 2-reach \newline with $\rho=1$ & 
$\Leftarrow$ & 2-reach \newline with $\rho=f+1$ &
$\Leftarrow$ & 3-reach \newline with $\rho=1$ 
\end{tabular}

Note that $P\Leftarrow Q$ means that condition $Q$ implies condition $P$ (or, condition $Q$ is stronger than condition $P$).

There are at least two different ways to construct the proofs. The approach we take uses only graph theoretic properties. An alternative proof method (which we do not use) would be based on the following observation: If any system that can solve problem A can also necessarily solve problem B, then a sufficient graph condition for problem A implies a sufficient graph condition for problem B. In any case, we do not pursue the latter approach here.

\subsection{Proof of ``1-reach with $\rho=1$ $\Leftarrow$ 1-reach with $\rho=f+1$''}
\label{a.1}

    Recall that we assume $f>0$. Now, for some set $F$ and nodes $u,v$, if $|reach_u(F)\cap reach_v(F)|\geq f+1$, then 
    $|reach_u(F)\cap reach_v(F)|\geq 1$.
    Therefore, the definition of the 1-reach condition trivially implies the following order:
    
    \hspace*{1in} 1-reach with $\rho=1$ $\Leftarrow$ 1-reach with $\rho=f+1$\\

\subsection{Proof of ``1-reach with $\rho=f+1$ $\Leftarrow$ 2-reach with $\rho=1$''}

 We will prove the contrapositive. Suppose that the graph $G(V,E)$ does not satisfy 1-reach with $\rho=f+1$. Then there exists $F\subset V$, $|F|\leq f$ and nodes $u,v\in V-F$ such that $|reach_u(F)\cap reach_v(F)|\leq f$. We consider two cases:
 \begin{itemize}
     \item (Case 1) $|reach_u(F)\cap reach_v(F)|=0$: Let us choose $F_1=F_2=F$. Then $u\in V-F_1$, $v\in V-F_2$, and $|reach_u(F_1)\cap reach_v(F_1)|=0$.
     This violates the 2-reach condition with $\rho=1$.
     
     \item (Case 2) $1\leq |reach_u(F)\cap reach_v(F)|\leq f$: 
     Let $X=reach_u(F)\cap reach_v(F)$. Then $X$ is non-empty. Also, by the definition of the reach sets $X\subseteq V-F$.
     
     Let us consider some node $w\in X$. If there is some node $z\in V-F-X$ in $reach_w(F)$, then $z$ has a path to $w$ in $G_F$. This, in conjunction with the fact that $w\in X$, implies that $z$ also has paths $u$ and $v$ both in $G_F$. Therefore, $z$ must be in $reach_u(F)$ and $reach_v(F)$ both. Hence $z\in X$, which is a contradiction with $z\in V-F-X$. Thus, we can infer that $reach_w(F)\subseteq X$ for $w\in X$, and $|reach_w(F)|\leq |X|\leq f$.
     
     Let us define $F_1=F$ and $F_2=reach_w(F)$. We can choose $F_2=reach_w(F)$ because $|reach_w(F)|\leq f$. Observe that $w\in X\subseteq V-F=V-F_1$. Consider any node $y\in V-F_2$. Then $reach_y(F_2)\subseteq V-F_2$, by the definition of a reach set. Since $F_2=reach_w(F)=reach_w(F_1)$, we get $reach_w(F_1)\cap reach_y(F_2)\subseteq F_2\cap (V-F_2)=\emptyset$. This violates the 2-reach condition with $\rho=1$.
 
  \end{itemize}
In both cases above, the 2-reach condition with $\rho=1$ is violated, proving the contrapositive.

\subsection{Proof of ``2-reach with $\rho=1$ $\Leftarrow$ 2-reach with $\rho=f+1$''}

Recall that we assume $f>0$. Now, for some sets $F_1$ and $F_2$, and nodes $u,v$, if $|reach_u(F_1)\cap reach_v(F_2)|\geq f+1$, then 
 $|reach_u(F_1)\cap reach_v(F_2)|\geq 1$.
    Therefore, the definition of the 2-reach condition trivially implies the following order:

\hspace*{1in} 2-reach with $\rho=1$ $\Leftarrow$ 2-reach with $\rho=f+1$\\

\subsection{Proof of ``2-reach with $\rho=f+1$ $\Leftarrow$ 3-reach with $\rho=1$''}

We will prove the contrapositive.
 To prove this let us assume that the given graph violates the 2-reach condition with $\rho=f+1$. Therefore, there exist sets $F_1,F_2\subseteq V$, $|F_1|,|F_2|\leq f$, and $u\in V-F_1$, $v\in V-F_2$ such that $|reach_u(F_1)\cap reach_v(F_2)|\leq f$. For brevity, let us define  $X=reach_u(F_1)\cap reach_v(F_2)$. Thus, $|X|\leq f$.

To simplify the presentation here, when a set contains a single element (e.g., set $\{u\}$), we will not write the curly brackets. For example, $X-\{u\}-\{v\}$ will be written as $X-u-v$ in the simplified notation, and $X\cup \{w\}$ will be written as $X\cup w$.
In other words, when using the simplified notation, $X-u=\{x~|~x\in X \mbox{~and~} x\neq u\}$.
 
Recall that we assume $f>0$. We will consider four cases below, and in each case, define new sets $\FF$, $\FF_1$ and $\FF_2$, such that the newly defined sets will satisfy the following properties: 
    \begin{itemize}
    \item (Property 1) $|\FF|, |\FF_1|,|\FF_2|\leq f$ 
    \item (Property 2) $u\in V-\FF-\FF_1$ and $v\in V-\FF-\FF_2$
     \item (Property 3) $F_1 \subseteq \FF\cup \FF_1$ and $F_2\subseteq \FF\cup \FF_2$
    \item (Property 4) $\FF\cup \FF_1\cup\FF_2=X\cup F_1\cup F_2$
    \item (Property 5) $reach_u(\FF\cup\FF_1)\cap reach_v(\FF\cup\FF_2)=\emptyset$
    \end{itemize}

Now we enumerate the four cases (and sub-cases, when needed) and define the new sets. Note that, by definition, $u \in reach_u(F_1)$ and $v \in reach_v(F_2)$. In each case, it is easy to see that the new sets satisfy properties 1 through 4 above. Later we will show that they satisfy property 5 as well. Recall that we assume $f>0$.
\begin{itemize}
    \item (Case 1) $u,v\not\in X$: Define $\FF=X$, $\FF_1=F_1$, $\FF_2=F_2$.
     \item (Case 2) $u \in X$ and $v\not\in X$: Thus, $u\neq v$. There are two possibilities in this case.
    \begin{list}{}{}
        \item[(i)] $F_2=\emptyset$: Define $\FF=X-u$, $\FF_1=F_1$, and $\FF_2=\{u\}$. 
        
        \item[(ii)] $F_2\neq\emptyset$: Since $u\in X$ and $X=reach_u(F_1)\cap reach_v(F_2)$, $u\in reach_v(F_2)$. Recall that $v\in V-F_2$. Also, since $u\in reach_v(F_2)$, $u\in V-F_2$. So $u,v\not\in F_2$, and since $F_2
        \neq\emptyset$, there must be some other node $w\in F_2$. Define $\FF=(X-u)\cup w$, $\FF_1=F_1$, $\FF_2=(F_2-w)\cup u$.
     
     \end{list}
  
    \item (Case 3) $u \not\in X$ and $v\in X$: This case is handled analogous to Case 2.

    \item (Case 4) $u \in X$ and $v\in X$: Here, similar to Case 2 above, we need to consider several possibilities.
    \begin{list}{}{}
        \item[(i)] $F_1=F_2=\emptyset$: Define $\FF=X-u-v$, $\FF_1=\{v\}$, and $\FF_2=\{u\}$.

        \item[(ii)] $F_1=\emptyset$ and $F_2\neq\emptyset$: Following an argument similar to that in Case 2(ii), there must be some node $w\in F_2$ that is distinct from $u$ and $v$. Define $\FF=(X-u-v)\cup w$, $\FF_1=\{v\}$, and $\FF_2=(F_2-w)\cup u$.
  
        \item [(iii)] $F_1\neq\emptyset$ and $F_2=\emptyset$: This is handled analogous to case (ii) above.
        
        \item [(iv)] $F_1\neq\emptyset$ and $F_2\neq\emptyset$: Then, similar to case (ii), there must be nodes $w_1\in F_1$ and $w_2\in F_2$ such that $w_1$ and $w_2$ are both distinct from $u$ and $v$ (however, it is possible that $w_1=w_2$). Define $\FF=(X-u-v)\cup w_1\cup w_2$, $\FF_1=(F_1-w_1)\cup v$ and $\FF_2=(F_2-w_2)\cup u$.
    \end{list}
    \end{itemize}

 It is straightforward to verify that the sets $\FF,\FF_1,\FF_2$ defined in each case satisfy properties 1, 2, 3 and 4 listed above. Using properties 1 through 4, we now show that these sets satisfy property 5 as well.

\begin{itemize}
    \item Properties 1 and 2 allow us to define  $reach_u(\FF\cup\FF_1)$ and $reach_v(\FF\cup\FF_2)$. By Property 3, $F_1\subseteq \FF\cup\FF_1$, therefore, $reach_u(\FF\cup\FF_1)\subseteq reach_u(F_1)$. Similarly, $reach_v(\FF\cup\FF_2)\subseteq reach_u(F_2)$. Therefore,
    \begin{eqnarray}
    reach_u(\FF\cup\FF_1)\cap reach_v(\FF\cup\FF_2)& \subseteq & reach_u(F_1)\cap reach_v(F_2) \\
    \Rightarrow reach_u(\FF\cup\FF_1)\cap reach_v(\FF\cup\FF_2) & \subseteq & X 
    \label{e_case4_1}
    \end{eqnarray}
    \item Now, by the definition of the reach sets, $reach_u(\FF\cup\FF_1)$ is disjoint with $\FF\cup \FF_1$. Similarly, $reach_v(\FF\cup\FF_2)$ is disjoint with $\FF\cup \FF_2$. This implies that $reach_u(\FF\cup\FF_1)\cap reach_v(\FF\cup\FF_2)$ is disjoint with $\FF\cup\FF_1\cup\FF_2$. Thus, by property 4, $reach_u(\FF\cup\FF_1)\cap reach_v(\FF\cup\FF_2)$ does not contain any nodes in $X\cup F_1\cup F_2$. Therefore,
    \begin{eqnarray}
    (reach_u(\FF\cup\FF_1)\cap reach_v(\FF\cup\FF_2))\cap X&=&\emptyset  \label{e_case4_2}
    \end{eqnarray}
\end{itemize}
Equations \ref{e_case4_1} and \ref{e_case4_2} together imply that
$reach_u(\FF\cup\FF_1)\cap reach_v(\FF\cup\FF_2)=\emptyset$, proving Property 5. Since properties 1 through 4 hold in all the four cases enumerated earlier, Property 5 is also proved in all the cases. Property 5 implies that 
$$|reach_u(\FF\cup\FF_1)\cap reach_v(\FF\cup\FF_2)|=0$$
Therefore, the 3-reach condition with $\rho=1$ is violated, proving the contrapositive.

\end{document}